\documentclass[11pt]{article}

\usepackage{amsmath}
\usepackage{amssymb}
\usepackage{amsthm}
\usepackage{amsfonts}
\usepackage{fullpage}
\usepackage{enumerate}
\usepackage{xcolor}
\usepackage{natbib}
\usepackage{multirow}
\usepackage[colorlinks = true, allcolors = blue]{hyperref}
\usepackage[linesnumbered, ruled, vlined, algo2e]{algorithm2e}
\usepackage{mathrsfs}
\usepackage{bbm}
\usepackage[capitalise]{cleveref}
\usepackage[title]{appendix}
\usepackage{subfigure,graphicx}
\usepackage{booktabs}
\usepackage{algorithm,algpseudocode,float}
\usepackage{threeparttable}

\theoremstyle{plain}
\newtheorem{thm}{Theorem}
\newtheorem{lem}{Lemma}

\newtheorem{coro}{Corollary}

\theoremstyle{definition}

\newtheorem{exmp}{Example}
\newtheorem{asmp}{Assumption}

\theoremstyle{remark}
\newtheorem{remark}{Remark}

\SetCommentSty{mycommfont}

\SetKwInput{KwInput}{Input}                % Set the Input
\SetKwInput{KwOutput}{Output} 
\algnewcommand\algorithmicinput{\textbf{Input:}}
\algnewcommand\algorithmicoutput{\textbf{Output:}}
\algnewcommand\Input{\item[\algorithmicinput]}%
\algnewcommand\Output{\item[\algorithmicoutput]}%

\newcommand{\N}{\mbox{$\textrm{\textup{N}}$}}
\newcommand{\RN}[1]{%
  \textup{\uppercase\expandafter{\romannumeral#1}}%
}

\newcommand{\bm}{\boldsymbol}
\newcommand{\vertiii}[1]{{\left\vert\kern-0.25ex\left\vert\kern-0.25ex\left\vert #1     \right\vert\kern-0.25ex\right\vert\kern-0.25ex\right\vert}}

\def\E{\mathbb{E}}
\def\var{\mathrm{Var}}

\def\cov{\mathrm{Cov}}

\begin{document}

% Specifying a running head, title, author(s), and affiliation
% in the manner required by the Journal.
% Note that the first argument of "\runningheads" specifies the
% "brief" running title (which will appear on the even numbered
% pages) and the second argument specifies the format in which
% the authors' names will appear on the odd numbered pages.
% Do *not* use "\corrauth" if there is only one author.
% *Do* use "\addressnum{1}" even if there is only one author!

\title{\large \bf Prediction De-Correlated Inference: A safe approach for post-prediction inference}

\author{}
%\email{example@mail.com}
\date{}
\maketitle
\vspace{-2cm}
\begin{center}
    Feng Gan\footnote{School of Statistics and Data Sciences, Nankai University, Tianjin, 300071, P.R.China\label{add1}}, Wanfeng Liang\footnote{School of Data Science and Artificial Intelligence, Dongbei University of Finance and Economics, Dalian, 116025, Liaoning, P.R.China\\  E-mail: liangwf@dufe.edu.cn}$^*$, Changliang Zou\textsuperscript{\ref{add1}}
\end{center}
%\date{\today}

%\email{example@mail.com}
%\keywords{asymptotic normality, black-box models; estimating equations; M-estimator; semi-supervised learning}

%}
% Specifying address(es) in the manner required by the Journal.

%Private Bag 92019, Auckland 1142, New Zealand\\
%\hspace*{1ex} Email: \texttt{r.turner@auckland.ac.nz}

% Note that the Journal requires that a paper must begin with a
% "Summary" not an "Abstract".  This is automatically taken care
% of by the anzsauth document style.  So even though the following
% says "\begin{abstract}" the heading "Summary" will appear in
% the processed version.
\begin{abstract}

In modern data analysis, it is common to use machine learning methods to predict outcomes on unlabeled datasets and then use these pseudo-outcomes in subsequent statistical inference. Inference in this setting is often called post-prediction inference. We propose a novel assumption-lean framework for statistical inference under post-prediction setting, called Prediction De-Correlated Inference (PDC). Our approach is safe, in the sense that PDC can automatically adapt to any black-box machine-learning model and consistently outperform the supervised counterparts. The PDC framework also offers easy extensibility for accommodating multiple predictive models. Both numerical results and real-world data analysis demonstrate the superiority of PDC over the state-of-the-art methods.
\end{abstract}

%\email[]{liangwf@dufe.edu.cn}
% Note that "keywords" should not include words and phrases
% that form part of the title of the paper.
% Also keywords (with the exception of proper names and
% certain abbreviations that conventionally appear all in
% capital letters) should *not* be capitalised.

% This shows how acknowledgements should appear in the Journal.
% Note that if you are acknowledging financial support, grant
% numbers should *not* usually be specified.
\section{Introduction}

The effectiveness of many modern techniques is heavily reliant on the availability of large amounts of labeled data, which can often be a bottleneck due to the costs associated with data annotation and the scarcity of labeled data in certain domains \citep{Chapelle2009}. Examples of this situation are ubiquitous, including electronic health records \citep{cai2022semitri}, expression quantitative trait loci studies \citep{michaelson2009detection}, public health \citep{khoury1999mortality}, among others.  The development of machine learning (ML) and deep learning techniques has led to the popular practice of using black-box models to predict missing outcomes and incorporating them into subsequent statistical analyses. While this approach offers certain advantages, it also introduces potential biases due to the reliance on the quality of the predicted outcomes.

In this article, we study a general semi-supervised M-estimation problem in a post-prediction setting. Suppose $(\text{Y},\mathbf{X}^{\top})^\top\in\mathcal{Y}\times\mathcal{X}$ is a $(p+1)$-dimensional random vector distributed according to $\Pr_{\text{Y},\mathbf{X}}$. We obtain a labeled dataset $\mathcal{D}_{L}=\{(y_i,\bm{x}_i^{\top})\in\mathbb{R}\times\mathbb{R}^{p},i=1,\dots,n\}$, where each sample is an i.i.d. realization from population $\Pr_{\text{Y},\mathbf{X}}$,  and an i.i.d. unlabeled dataset $\mathcal{D}_{U}=\{\bm{x}_i\in\mathbb{R}^{p},i=n+1,\dots,n+N\}$ from $\Pr_{\mathbf{X}}$, the marginal distribution of $\mathbf{X}$. Besides the presence of additional unlabeled data $\mathcal{D}_U$,  we also assume that a predictive model $\mu(\cdot)$ independent of $\mathcal{D}_L$ and $\mathcal{D}_U$ is available. Formally, we aim to  perform inference on the parameter $\theta^*\in\Theta$, defined as the solution to the estimating equation
\begin{equation}\label{equ:score}
	\E\{\bm{s}(\text{Y},\mathbf{X},\theta)\}=0,
\end{equation}
where $\bm{s}(\cdot,\cdot,\cdot):\mathcal{Y}\times\mathcal{X}\times\Theta\rightarrow \mathbb{R}^d$ is a pre-specified estimating function, $\E(\cdot)$ denotes the expectation with respect to all the stochastic terms, $\Theta\subset\mathbb{R}^d$ and $d$ is some positive integer. If $\bm{s}(y,\bm{x},\theta)=\nabla_{\theta} l(y,\bm{x},\theta)$ for some convex function $l(\cdot,\cdot,\cdot):\mathcal{Y}\times\mathcal{X}\times\Theta\rightarrow \mathbb{R}$, then $\theta^*$ can also be defined as $\theta^*\in\underset{\theta\in\Theta}{\arg\min}\ \E \{l(\text{Y},\mathbf{X},\theta)\}$. In this article, we do not assume that this property holds, that is to say, $\bm{s}(\cdot,\cdot,\theta)$ may not be written in the form of the gradient of a certain loss function.

Our goal is to propose an estimator that is asymptotically more efficient than the supervised counterpart $\hat{\theta}_{sup}$ which solves the estimating equation $n^{-1}\sum_{i=1}^n \bm{s}(y_i,\bm{x}_i,\theta)=0$ without using any information from $\mathcal{D}_U$ and $\mu(\cdot)$. Under some regular conditions, classic statistical theory \citep{vaart_1998} implies that the asymptotic covariance of  $\hat{\theta}_{sup}$ essentially depends on the properties of the estimating function at point $\theta^*$, 
\begin{gather}\label{equ:sup asymp dist}
	\sqrt{n}(\hat{\theta}_{sup}-\theta^*)\overset{p}{\rightarrow} \N_d(\bm{0},\bm{H}^{-1}{\cov}\{\bm{s}(\text{Y},\mathbf{X},\theta^*)\}\{\bm{H}^{\top}\}^{-1}), 
\end{gather}
where $\bm{H}=\nabla_{\theta^{\top}}\E\{\bm{s}(\text{Y},\mathbf{X},\theta^*)\}$. Thus, to obtain a more statistically efficient estimator, it is reasonable to adjust the empirical estimating function by incorporating additional information from $\mathcal{D}_U$ and $\mu(\cdot)$.  A straightforward approach is to impute the missing labels in $D_U$ via $\mu(\cdot)$  and then treat $\mathcal{D}_L\cup\{(\mu(\bm{x}_i),\bm{x}_i),i=n+1,\dots,n+N\}$ as the true dataset. Then a supervised approach can be applied to this pseudo dataset to obtain the desired estimator. The efficiency of this strategy clearly depends on the performance of  $\mu(\cdot)$. An inappropriate $\mu(\cdot)$ can result in an estimator that is less efficient than $\hat{\theta}_{sup}$.  \cite{wang2020methods} suggest applying a parametric model to capture the relationship between true labels and predictive model $\mu(\cdot)$ and then using this relationship to correct the predictions on the unlabeled data. It is required that the relationship between  can be characterized by a simple parametric model, which is often not met in practice. Recently, \cite{angelopoulos2023prediction} proposed a general framework called Prediction-Powered Inference (PPI). The PPI rectifies the predicted estimating function and can provide valid confidence intervals even using black-box ML models. It works with the following bias-corrected estimating function:
\begin{gather*}\label{equ:deb}
	\hat{S}_{ppi}(\theta)
	=\Big\{\underbrace{\frac{1}{n}\sum_{i=1}^n \bm{s}(y_i,\bm{x}_i,\theta)-\frac{1}{n}\sum_{i=1}^n\bm{s}(\mu(\bm{x}_i),\bm{x}_i,\theta)}_{debiasing\ term}+\frac{1}{N}\sum_{i=n+1}^{n+N}\bm{s}(\mu(\bm{x}_i),\bm{x}_i,\theta)\Big\}.
\end{gather*}
This idea has also been studied in the literature of semi-supervised learning which mainly focused on some special cases of M-estimation problems such as mean estimation \citep{Zhang2019ssmean,zhang2022high}, linear regression \citep{chakrabortty2018,azriel2022semi}, and quantile estimation \citep{chakrabortty2022semi}. Although the estimator  $\hat{\theta}_{ppi}\in\{\theta: \hat{S}_{ppi}(\theta)=0\}$ is asymptotically consistent, it is not {\it always} asymptotically more efficient than the supervised estimator $\hat{\theta}_{sup}$. To be more specific, let $\bm{s}_{\theta}=\bm{s}(\text{Y},\mathbf{X},\theta)$ and $\bm{s}_{\theta}^{\mu}=\bm{s}(\mu(\mathbf{X}),\mathbf{X},\theta)$, it is required that the matrix
\[
\cov(\bm{s}_{\theta^*})-\cov(\bm{s}_{\theta^*}-\bm{s}^{\mu}_{\theta^*})-\frac{n}{N}\cov(\bm{s}^{\mu}_{\theta^*})
\]
is positive semi-definite, which is not always satisfied if the $\mu(\cdot)$ is not ‘good’ enough or $N$ is not large enough.  Both \cite{schmutz2023dont} and \cite{angelopoulos2023ppi++} developed a lightweight modification of PPI to improve efficiency. However, their methods only guarantee a reduction of the trace of the asymptotic covariance matrix, and thus they are still not always superior to $\hat{\theta}_{sup}$.

More recently, \cite{song2023general} studied the general M-estimation problem and proposed a new estimator via projection technique. The projection basis is some polynomial functions of $\mathbf{X}$.% and the order of polynomial is determined by the $\text{GBIC}_p$ procedure \citep{lv2014model}.  
Their method produces asymptotically more efficient estimators for the target parameters than the supervised counterpart. However, one often can have access to some additional information in practice, e.g., surrogate variables in electronic health records \citep{hou2023surrogate} and genetic variants from first-degree relatives in case-control genetic association studies \citep{liu2017case}. Thus, it is of interest to develop methods that are able to utilize such information which is summarized as the pre-trained model $\mu(\cdot)$.

%\citep{zhang2022high,chakrabortty2022semi} and missing data literature \citep{robins1995analysis,robins1995semiparametric}. Previous semi-supervised literature 

%Besides post-prediction inference, another closely related field is semi-supervised learning.  

\subsection{Our contributions}
In this article, we propose a new framework called \emph{Prediction De-Correlated inference}  (PDC) to incorporate the unlabeled data and predictive ML models for the M-estimation problem. Our contributions lie in three folds. 
\begin{itemize}
	\item The estimator based on the PDC framework is `safe', in the sense that even when the black-box ML model makes a poor prediction, the PDC estimator is always asymptotically no worse than its supervised counterpart. 
	\item Our work has unique advantages compared with several existing  works. If  $\theta^*$ is multi-dimensional, one may want to perform inference on $\bm{c}^{\top}\theta^*$ for some vector $\bm{c}\in\mathbb{R}^d$. The proposed PDC  also provides safe estimates under this scenario. Existing  works \citep{schmutz2023dont,angelopoulos2023ppi++,miao2023assumption} are all constrained to  element-wise variance reduction or  trace reduction. In contrast, our method takes into account the correlation between the components of the estimator, yielding a more reliable procedure. For detailed numerical comparisons, see Section \ref{sec:experiments}.
	\item  The proposed PDC framework can be easily extended to  accommodate multiple predictive models; see Subsection \ref{subsec:PDC} for more detailed descriptions.
\end{itemize}

\subsection{Organization and notation}

The rest of this article is organized as follows. In Section \ref{sec:main method}, we describe the proposed Prediction De-Correlated framework. In Section \ref{sec:theorems}, we establish the asymptotic properties for the PDC estimator. Numerical studies and the real data analysis are summarized in Section \ref{sec:experiments}. Concluding remarks and discussions for future work are given in Section \ref{sec:discussion}.  All 
technical proofs can be found in the Appendix.

For any vector $\bm{c}$, $\Vert\bm{c}\Vert$ denotes its $L_2$ norm. For any matrix $\bm{C}$, $\Vert\bm{C}\Vert$ denotes its spectral norm. The identity matrix is denoted by $\bm{I}$. We use $\E$ and $\cov$ to denote the expectation and covariance with respect to all the stochastic terms,  respectively.
We add subscript $n$ to $\E(\cdot)$ and $\cov(\cdot)$ to denote the sample version of the population mean and covariance respectively, e.g., for any measurable functions $g(\cdot)$ and $h(\cdot)$ defined in $\mathcal{X}$, $\E_n (g)=\frac{1}{n}\sum_{i=1}^n g(\bm{x}_i)$, $\cov_n (g, h)=\frac{1}{n}\sum_{i=1}^n \{g(\bm{x}_i)-\E_n (g)\}\{h(\bm{x}_i)-\E_n (h)\}^{\top}$. Similarly, we add subscript $N$ to $\E(\cdot)$ to denote the sample version of the population mean based on the unlabeled data, e.g., $\E_N (g)=\frac{1}{N}\sum_{i=n+1}^{n+N}g(\bm{x}_i)$. A vector with subscript $(i)$ means the $i$-th element of that vector. A matrix with subscript $(ij)$  means the $(i,j)$-th entry of that matrix.

\section{Methodology}\label{sec:main method}

\subsection{Prediction de-correlated estimator}\label{subsec:PDC}

Our starting point is to utilize unlabeled data $\mathcal{D}_U$ and predictive model $\mu(\cdot)$ to modify the empirical estimating function, with the aim of reducing the covariance (in the matrix sense) at $\theta^*$. Define the predictive estimating function $\bm{f}(\cdot,\cdot):\mathcal{X}\times\Theta\rightarrow\mathbb{R}^{q}$, where $q$ is some positive integer which can be different from $d$. We denote $\bm{s}_1=\E_n\{\bm{s}(\text{Y},\mathbf{X},\theta)\},\bm{s}_2=\E_n\{\bm{f}(\mathbf{X},\theta)\},\bm{s}_3=\E_N\{\bm{f}(\mathbf{X},\theta)\}$. The $\bm{s}_1$ is the empirical estimating function that we use in the supervised inference procedure,  which only involves the information from $\mathcal{D}_L$. The $\bm{s}_2$ and $\bm{s}_3$ are the empirical predictive estimating functions,  which utilize the information from the predictive model on $\mathcal{D}_L$ and $\mathcal{D}_U$, respectively. Rather than using $\bm{s}_1$ and $\bm{s}_2$ to debias $\bm{s}_3$, we search for  a more efficient way to make use of the additional information contained in $\mathcal{D}_{U}$ and $\mu(\cdot)$. We consider solving the following modified estimating equation:
\begin{gather}\label{equ:modi-score-abc}
	\bm{s}_1+\bm{A}\bm{s}_2+\bm{B}\bm{s}_3=\bm{0},
\end{gather}
where $\bm{A},\bm{B}$ are $d\times q$ matrices needed to be specified. Taking the unbiasedness into account, we rewrite it as 
\begin{gather*}\label{equ:pdc modi-score}
	\bm{s}_1+\bm{T}(\bm{s}_2-\bm{s}_3)=\bm{0},
\end{gather*}
where $\bm{T}$ is a $d\times q$ matrix needed to be specified. 

Our goal has shifted to finding an appropriate $\bm{T}$ such that the resulting estimator is more efficient than the supervised counterpart. 
%Notice that the nuisance parameter $\bm{T}$ is a matrix when $d>1$. 
%It is troublesome to find all the matrices $\bm{T}$ such that the resulting estimator is no less efficient than $\hat{\theta}_{sup}$. So 
We propose an intuitive choice for $\bm{T}$. Let $\bm{T}=\gamma\widehat{\bm{T}}(\theta)$, where  
\begin{equation}\label{equ:hat T}
	\widehat{\bm{T}}(\theta)=\cov_n(\bm{s}_{\theta},\bm{f}_{\theta})\{\cov_n(\bm{f}_{\theta})\}^{-1}
\end{equation}
is the sample version of $\bm{T}^*(\theta)=\cov(\bm{s}_{\theta},\bm{f}_{\theta})\{\cov(\bm{f}_{\theta})\}^{-1}$, $\bm{s}_{\theta}= \bm{s}(\text{Y},\mathbf{X},\theta)$, $\bm{f}_{\theta}=\bm{f}(\mathbf{X},\theta)$ and $\gamma$ is a constant needed to be specified. 
The motivation behind the choice of $\bm{T}$ is as follows. First, $\bm{T}^*(\theta^*)$ is the multivariate least square coefficient matrix, thus $\bm{T}^*(\theta^*)\{\bm{f}_{\theta^*}-\E(\bm{f}_{\theta^*})\}$ is the $L_2$ projection of $\bm{s}_{\theta^*}$ into the linear space spanned by the components of $\bm{f}_{\theta^*}-\E(\bm{f}_{\theta^*})$. Hence $S(\text{Y},\mathbf{X},\theta^*)=\bm{s}_{\theta^*}-\bm{T}^*(\theta^*)\{\bm{f}_{\theta^*}-\E(\bm{f}_{\theta^*})\}$ is the decorrelated estimating function which is orthogonal  to $\bm{f}_{\theta^*}-\E(\bm{f}_{\theta^*})$.  The decorrelated estimating function certainly has a smaller covariance (in the matrix sense) than the original one. The larger the space spanned by $\bm{f}_{\theta^*}-\E(\bm{f}_{\theta^*})$, the smaller the covariance of the decorrelated estimating function we will obtain. 

The choice of the predictive estimating function  $\bm{f}(\cdot,\cdot)$ can be diverse. In practice, we can set $\bm{f}(\bm{x},\theta)$ as $\bm{s}(\mu(\bm{x}),\bm{x},\theta)$. Notice that the dimension of $\bm{f}(\cdot,\cdot)$ can be different from that of $\bm{s}(\cdot,\cdot,\cdot)$. We can augment a predictive estimating function with another predictive estimating function, for example,
\begin{gather}\label{equ:combine}
	\bm{f}(\bm{x},\theta)=(\bm{f}_1^\top(\bm{x},\theta),\bm{f}_2^\top(\bm{x},\theta))^\top=(\bm{s}^{\top}(\mu_1(\bm{x}),\bm{x},\theta),\bm{s}^{\top}(\mu_2(\bm{x}),\bm{x},\theta))^{\top},
\end{gather}
where $\mu_1(\cdot)$ and $\mu_2(\cdot)$ are two different predictive models, both of which are independent of $\mathcal{D}_L$, $\mathcal{D}_U$. Because $\mu(\cdot)$ is independent of $\mathcal{D}_L$ and $\mathcal{D}_U$, the elements in the set $F(\theta)=\{\bm{f}(\bm{x}_i,\theta): i=1,\dots,n,\dots,n+N\}$ are exchangeable. Thus, in principle, we can use any subset of $F$ to estimate $\E\{\bm{f}(\mathbf{X},\theta)\}$ and reserve the rest for projection. In order to increase the correlation between $\bm{s}_1$ and the projection part, we use the unlabeled data to estimate $\E\{\bm{f}(\mathbf{X},\theta)\}$. Finally, $\gamma$ is  used to adjust for different sample sizes  used in $\bm{s}_1$ and the projection step. Thus, we can estimate $\E \{S(\text{Y},\mathbf{X},\theta)\}$ by
\begin{equation}\label{equ:hat S}
	\hat{S}(\theta)=\frac{1}{n}\sum_{i=1}^n\Big[\bm{s}(y_i,\bm{x}_i,\theta)+\gamma\widehat{\bm{T}}(\theta)\big\{\bm{f}(\bm{x}_i,\theta)-\frac{1}{N}\sum_{i=n+1}^{n+N}\bm{f}(\bm{x}_i,\theta)\big\}\Big].
\end{equation}

Solving  $\hat{S}(\theta)=0$ directly could result in multiple solutions. So we use a one-step estimator instead.   We formalize our one-step Prediction De-Correlated  procedure in the following algorithm.
\begin{algorithm}
	\caption{Calculate the one-step Prediction De-Correlated  estimator $\hat{\theta}_{pdc,1}$}
	%\begin{flushleft}

	%\end{flushleft}
	\begin{algorithmic}
		\Input Labeled data $\mathcal{D}_L$, unlabeled data $\mathcal{D}_U$, initial estimator $\hat{\theta}_0$, predictive estimating function $\bm{f}(\cdot,\cdot)$.
		\State{Calculate sample version of the multivariate least square coefficient matrix  $\widehat{\bm{T}}(\hat{\theta}_0)$ in \eqref{equ:hat T}.}
		\State{Calculate the estimated  prediction de-correlated  estimating function  $\hat{S}(\hat{\theta}_0)$ in \eqref{equ:hat S}.}
		\State{Calculate the one-step prediction de-correlated estimator $\hat{\theta}_{pdc,1}$,
			\begin{gather*}
				\hat{\theta}_{pdc,1}=\hat{\theta}_0-\widehat{\bm{H}}^{-1}\hat{S}(\hat{\theta}_0).
			\end{gather*}
			where  $\widehat{\bm{H}}$ is an estimator of $\bm{H}=\nabla_{\theta^{\top}}\E\{\bm{s}(\text{Y},\mathbf{X},\theta^*)\}$. 
		}
		\Output One-step Prediction De-Correlated  estimator $\hat{\theta}_{pdc,1}$.
	\end{algorithmic}\label{alg:one-step}
	
	%\begin{flushleft}
	
	%\end{flushleft}
\end{algorithm}

 If an additional dataset for training $\mu(\cdot)$ is obtainable and an initial estimator of $\theta^*$ is available, an alternative procedure is to directly utilize this dataset to train a model on $(\bm{s}(\text{Y},\mathbf{X},\theta^*),\mathbf{X})$. Both the Algorithm \ref{alg:one-step} and such variant can be applied iteratively. We conclude this subsection with two examples to further clarify our methodology.

\begin{exmp}[Mean estimation]\label{exmp:mean estimation}
	If $s(y,\bm{x},\theta)=\theta-y$, then we are actually estimating $\E (\text{Y})$. In this case, $\bm{H}$ is fixed at $1$, therefore we do not need to estimate it. If we choose $\hat{\theta}_0=\hat{\theta}_{sup}=\E_n (\text{Y})$ and set $\bm{f}(\bm{x},\theta)=\bm{x}$, then 
 \[
 \hat{\theta}_{pdc,1}=\E_n(\text{Y})-\gamma\cov_n(\text{Y},\mathbf{X})\{\cov_n(\mathbf{X})\}^{-1}\{\E_n (\mathbf{X})-\E_N (\mathbf{X})\}.
 \] 
 Now $\hat{\theta}_{pdc,1}$  is just the mean estimator proposed by \cite{Zhang2019ssmean}. With an additional pre-trained model $\mu(\cdot)$, we can set $\bm{f}(\bm{x},\theta)=\theta-\mu(\bm{x})$, and then $\hat{\theta}_{pdc,1}$ becomes
 \[
 \hat{\theta}_{pdc,1}=\E_n(\text{Y})-\gamma\cov_n\{\text{Y},\mu(\mathbf{X})\}[\cov_n\{\mu(\mathbf{X})\}]^{-1}[\E_n\{\mu(\mathbf{X})\}-\E_N\{\mu(\mathbf{X})\}].
 \]
\end{exmp}

\begin{exmp}[Generalized Linear Models]\label{exmp:glm}
	In the generalized linear models, the loss function is $l(y,\bm{x},\theta)=-y\bm{x}^{\top}\theta+b(\bm{x}^{\top}\theta)$. This corresponds to the estimating function $\bm{s}(y,\bm{x},\theta)=-y\bm{x}+b'(\bm{x}^{\top}\theta)\bm{x}$. We can estimate $\bm{H}$ by $\widehat{\bm{H}}=\E_n \{b''(\mathbf{X}^{\top}\hat{\theta}_{0})\mathbf{X}\mathbf{X}^{\top}\}$. In linear model, $\bm{s}(y,\bm{x},\theta)=(\bm{x}^{\top}\theta-y)\bm{x}$, $\widehat{\bm{H}}=\E_n(\mathbf{X}\mathbf{X}^{\top})$. In logistic model, 
 \[
 s(y,\bm{x},\theta)=-y\bm{x}+\frac{\exp({\bm{x}^{\top}\theta})}{1+\exp({\bm{x}^{\top}\theta})}\bm{x},
 \]
 and
 $\widehat{\bm{H}}=\E_n\{\frac{\exp({\mathbf{X}^{\top}\hat{\theta}_0})}{(1+\exp({\mathbf{X}^{\top}\hat{\theta}_0}))^2}\mathbf{X}\mathbf{X}^{\top}\}.$
 
\end{exmp}

\begin{remark}
    We notice an independent concurrent work of this article, POst-Prediction Inference (POP-Inf) \citep{miao2023assumption}. They also aim to gain efficiency by modifying the estimating function. But in their settings, both $\mathcal{D}_L$ and $\mathcal{D}_U$ contain an auxiliary random vector $\mathbf{Z}$ that is predictive of $\text{Y}$, that is $\mathcal{D}_L=\{(y_i,\bm{x}_i,\bm{z}_i):i=1,\dots,n\}$, $\mathcal{D}_U=\{(\bm{x}_i,\bm{z}_i):i=n+1,\dots,n+N\}$, and $\mu(Z)$ produces predictions for $\text{Y}$. If $\mathbf{Z}=\mathbf{X}$, then their settings are the same as ours. In this case, their modified estimating function is
\begin{gather*}
	\frac{1}{n}\sum_{i=1}^n\bm{s}(y_i,\bm{x}_i,\theta)+\omega\odot\big\{\frac{1}{n}\sum_{i=1}^n \bm{s}(\mu(\bm{x}_i),\bm{x}_i,\theta)-\frac{1}{N}\sum_{i=n+1}^{n+N}\bm{s}(\mu(\bm{x}_i),\bm{x}_i,\theta)\big\},
\end{gather*}
where $\bm{\omega}$ is a weighting vector and $\odot$ denotes the element-wise multiplication. This corresponding to $\bm{T}=\text{diag}(\bm{\omega})$ in \eqref{equ:pdc modi-score}. They search for the best  $\hat{\bm{\omega}}$ by reducing the variance of each element in the estimator  successively until the algorithm converges. However, they do not give a closed form for the limit of $\hat{\bm{\omega}}$ with respect to the iterative times. So the asymptotic distribution of the limit of their estimator of $\theta^*$ is also unclear. Moreover, POP-Inf can only utilize one predictive model, while our PDC procedure can adapt to multiple predictive models.
\end{remark}

\section{Theoretical results}\label{sec:theorems}

In this section, we explore the theoretical properties of the  PDC procedure. In Subsection \ref{subsec:asymp normality}, we establish the asymptotic normality for the PDC estimator. In Subsection \ref{subsec:relationship with song}, we show the relationship between the PDC and semi-supervised methods, and extend the PDC to accommodate semi-supervised setting via cross-fitting strategy. In Subsection \ref{subsec:random model}, we further extend the PDC to accommodate the stochastic predictive model setting.

\subsection{Asymptotic normality}\label{subsec:asymp normality}
To proceed with the theoretical analysis, we need several assumptions.
\begin{asmp}\label{asmp:theta is interior}
	$\theta^*$ is an interior point of $\theta$, where $\theta$ is a compact subset of $\mathbb{R}^d$.
\end{asmp}

\begin{asmp}\label{asmp:consist init}
	$\hat{\theta}_0$ is a $\sqrt{n}$-consistent estimator of $\theta^*$, that is  $\Vert\hat{\theta}_0-\theta^*\Vert=O_p(n^{-1/2})$, $\widehat{\bm{H}}$ is a consistent estimator of $\bm{H}$ such that $\Vert\widehat{\bm{H}}-\bm{H}\Vert=o_p(1)$.
\end{asmp}

\begin{asmp}\label{asmp:score is L2 bounded}
	$\E\{\Vert\bm{s}(\text{Y},\mathbf{X},\theta^*)\Vert^2\}<\infty$, $\E\{\Vert\bm{f}(\mathbf{X},\theta^*)\Vert^2\}<\infty$.
\end{asmp}

\begin{asmp}\label{asmp:covf and H is nonsingular}
	$\cov(\bm{f}_{\theta^*})$ and $\bm{H}$ are nonsingular.
\end{asmp}

\begin{asmp}\label{asmp:smooth score}
	For every $\theta_1$ and $\theta_2$ in a neighborhood of $\theta^*$, there exist measurable functions $\dot{s}(\cdot,\cdot)$ and $\dot{f}(\cdot)$ with 
	\begin{gather*}
		\E\{\dot{s}^2(\text{Y},\mathbf{X})\}<\infty,\quad\E\{\dot{f}^2(\mathbf{X})\}<\infty
	\end{gather*}
	such that estimating function $\bm{s}(\cdot,\cdot,\cdot)$ and predictive estimating function $\bm{f}(\cdot,\cdot)$ satisfy the following smoothness conditions:
	\begin{gather*}
		\Vert\bm{s}(y,\bm{x},\theta_1)-\bm{s}(y,\bm{x},\theta_2)\Vert\le\dot{s}(y,\bm{x})\Vert\theta_1-\theta_2\Vert,\\
		\Vert\bm{f}(\bm{x},\theta_1)-\bm{f}(\bm{x},\theta_2)\Vert\le\dot{f}(\bm{x})\Vert\theta_1-\theta_2\Vert
	\end{gather*}
 for any $(y,\bm{x})\in\mathcal{Y}\times\mathcal{X}$.
\end{asmp}

Assumption \ref{asmp:theta is interior} ensures that $\theta^*$ is a reasonable target to estimate. Assumption \ref{asmp:consist init} ensures $\hat{\theta}_0$ is a good initial estimator, and can be satisfied by supervised estimator $\hat{\theta}_{sup}$.  Assumption \ref{asmp:score is L2 bounded}  enables the use of the central limit theorem.   Assumption \ref{asmp:covf and H is nonsingular} is a regularity condition. Assumption \ref{asmp:smooth score} is a smoothness condition to guarantee the Donsker property. Similar assumptions can be found in \cite{vaart_1998}. The following Theorem \ref{thm:one-step pdc} establishes the asymptotic normality of $\hat{\theta}_{pdc,1}$.

\begin{thm}\label{thm:one-step pdc}
	Under Assumptions \ref{asmp:theta is interior}-\ref{asmp:smooth score}, assuming that $N/(n+N)\rightarrow \eta$ for some constant $\eta\in(0,1]$,  then for any fixed $\gamma\in\mathbb{R}$ we have
	\begin{align}\label{equ:pdc aymp normal}
		\sqrt{n}(\hat{\theta}_{pdc,1}-\theta^*)\overset{D}{\rightarrow} N_d(\bm{0}, \bm{H}^{-1}\bm{\Gamma}(\gamma)\{\bm{H}^{\top}\}^{-1}),
	\end{align}
	where  $\bm{\Gamma}(\gamma)=\cov\{\bm{s}_{\theta^*}\}+(\gamma^2/\eta+2\gamma)\cov(\bm{s}_{\theta^*},\bm{f}_{\theta^*})\{\cov(\bm{f}_{\theta^*})\}^{-1}\cov(\bm{f}_{\theta^*},\bm{s}_{\theta^*})$.
\end{thm}

\begin{remark}\label{rmk: gamma analysis}
	Recall \eqref{equ:sup asymp dist}, we know that $\sqrt{n}(\hat{\theta}_{sup}-\theta^*)\overset{D}{\rightarrow} \N_d(\bm{0},\bm{H}^{-1}\cov(\bm{s}_{\theta^*})\{\bm{H}^{\top}\}^{-1})$. Since the matrix 
	\[
	\cov(\bm{s}_{\theta^*},\bm{f}_{\theta^*})\{\cov(\bm{f}_{\theta^*})\}^{-1}\cov(\bm{f}_{\theta^*},\bm{s}_{\theta^*})
	\]
	is a positive semi-definite matrix, we need $(\gamma^2/\eta+2\gamma)\le0$ to obtain an estimator that is asymptotically more  efficient than $\hat{\theta}_{sup}$. Simple mathematical operations give that:
	\begin{itemize}
		\item When $\gamma\in \{-2\eta,0\}$, $\hat{\theta}_{pdc}$ is asymptotically as efficient as  $\hat{\theta}_{sup}$.
		\item When $\gamma\in(-2\eta,0)$:  
		\begin{itemize}
			\item If $\cov(\bm{s}_{\theta^*},\bm{f}_{\theta^*})\cov(\bm{f}_{\theta^*},\bm{s}_{\theta^*})$ is a positive definite matrix, then   $\hat{\theta}_{pdc}$ is asymptotically more efficient than $\hat{\theta}_{sup}$. 
			\item If $\bm{c}^{\top}\cov(\bm{s}_{\theta^*},\bm{f}_{\theta^*})\neq \bm{0}_q$ for a vector $\bm{c}\in\mathbb{R}^d$,   then $\bm{c}^{\top}\hat{\theta}_{pdc,1}$ is more asymptotically more efficient than $\bm{c}^{\top}\hat{\theta}_{sup}$.  Notice that we have assumed $\eta \neq 0$ in Theorem \ref{thm:one-step pdc}. Otherwise we may encounter a case where no PDC estimators are asymptotically more efficient than $\hat{\theta}_{sup}$.
		\end{itemize}
		\item The optimal $\gamma$ is $\gamma_{opt}=-\eta$. In this case,
		\begin{gather*}
			\bm{\Gamma}(\gamma_{opt})=\cov\{\bm{s}_{\theta^*}\}-\eta\cdot \cov(\bm{s}_{\theta^*},\bm{f}_{\theta^*})\{\cov(\bm{f}_{\theta^*})\}^{-1}\cov(\bm{f}_{\theta^*},\bm{s}_{\theta^*}).
		\end{gather*}
	\end{itemize}
	
\end{remark}

In order to have a more intuitive understanding of Theorem \ref{thm:one-step pdc}, we consider a 1-dimensional case, that is, $d=1$. For $\gamma\in(-2\eta,0)$, we have $\bm{\Gamma}(\gamma)-\bm{\Gamma}(0)\le 0$ with equality holds only when $\cov(\bm{s}_{\theta^*},\bm{f}_{\theta^*})=\bm{0}$. As the linear space spanned by $\bm{f}_{\theta^*}$'s  components is uncorrelated with $\bm{s}_{\theta^*}$, expecting an efficiency gain through $\bm{f}_{\theta^*}$ is unrealistic. Moreover, in this case, the estimation of $\bm{T}^*(\theta^*)$, specifically, $\widehat{\bm{T}}(\hat{\theta}_0)$, should well approach $\bm{0}$ provided that a good initial estimator $\hat{\theta}_0$ is used. This data-adaptive weight selection property is the key reason why the PDC estimator is always no worse than the supervised counterpart.

Since we know the optimal $\gamma_{opt}$, we can estimate it in practice. A natural choice is $\gamma_n=-N/(n+N)$. The following corollary provides the asymptotic distribution of $\hat{\theta}_{pdc,1}$ when we set $\gamma=\gamma_n$.
\begin{coro}\label{coro:estimated optimal gamma}
	Under the conditions of Theorem \ref{thm:one-step pdc}, if we choose $\gamma=-N/(n+N)$, then
	\begin{equation*}
		\sqrt{n}(\hat{\theta}_{pdc,1}-\theta^*)\overset{D}{\rightarrow} N_d(\bm{0}, \bm{H}^{-1}\bm{\Gamma}(\gamma_{opt})\{\bm{H}^{\top}\}^{-1}).
	\end{equation*}
\end{coro}

\subsection{Relationship with semi-supervised literature}\label{subsec:relationship with song}

If we choose  $\bm{f}(\bm{x},\theta)$ as some deterministic function $Z(\bm{x})$, $\gamma=-N/(n+N)$, and denote 
\begin{gather*}
	\tilde{Z}(\bm{x})= (1,Z(\bm{x})^{\top})^{\top}, \bm{U}_n(\theta)=\E_n\{\bm{s}(\text{Y},\mathbf{X},\theta)\tilde{Z}(\mathbf{X})^{\top}\}[\E_n\{\tilde{Z}(\mathbf{X})\tilde{Z}(\mathbf{X})^{\top}\}]^{-1},
\end{gather*}
then  
\begin{equation*}
	\hat{S}(\theta)=\bm{U}_n(\theta)\frac{1}{n+N}\sum_{i=1}^{n+N}\tilde{Z}(\bm{x}_i).
\end{equation*}
This is equivalent to using the method of \cite{Zhang2019ssmean} to estimate each element of $\E\{\bm{s}(\text{Y},\mathbf{X},\theta)\}$ based on labeled data $\{(\bm{s}(y_i,\bm{x}_i,\theta),\tilde{Z}(\bm{x}_i)),i=1,\dots,n\}$ and unlabeled data $\{\tilde{Z}(\bm{x}_i),i=n+1,\dots,n+N\}$.
By the projection property,  $\hat{S}(\theta)$ can also be written as 
\begin{gather*}
	\hat{S}(\theta)=\frac{1}{n+N}\sum_{i=1}^n\bm{s}(y_i,\bm{x}_i,\theta)+\frac{1}{n+N}\sum_{i=n+1}^{n+N}\bm{U}_n(\theta)\tilde{Z}(\bm{x}_i).
\end{gather*}
Searching for the zero point of this $\hat{S}(\theta)$ is equivalent to finding the minimizer of the following modified  risk function
\begin{gather*}
	\hat{R}(\theta)=\frac{1}{n+N}\sum_{i=1}^n l(y_i,\bm{x}_i,\theta)+\frac{1}{n+N}\sum_{i=n+1}^{n+N}\widehat{\bm{\beta}}_l(\theta)^{\top}\tilde{Z}(\bm{x}_i),
\end{gather*}
where $\widehat{\bm{\beta}}_l(\theta)=[\E_n \{\tilde{Z}(\mathbf{X})\tilde{Z}(\mathbf{X})^{\top}\}]^{-1}\E_n \{\tilde{Z}(\mathbf{X}) l(\text{Y},\mathbf{X},\theta)\}$.
This is the weighted loss function with the optimal weight in \cite{song2023general}. Thus, our PDC procedure can be seen as an extension of the procedure in \cite{song2023general} to accommodate the pre-trained ML models. 

In the classic semi-supervised setting where a pre-trained predictive model is unavailable, the PDC procedure can still be applied in a cross-fitting form \citep{cherno2018, zhang2022high, chakrabortty2022semi, zrnic2024cross}.  We split $\mathcal{D}_L$ into $K$ folds $\mathcal{D}_{L}^{(1)}, \dots, \mathcal{D}_{L}^{(K)}$, each of size $n/K$, and split $\mathcal{D}_U$ into $K$ folds $\mathcal{D}_{U}^{(1)}, \dots, \mathcal{D}_{U}^{(K)}$, each of size $N/K$, where we have assumed $n$ and $N$ are divisible by $K$ for simplicity. We utilize all folds but $\mathcal{D}_{L}^{(j)}$ to train the predictive estimating function $\hat{\bm{f}}^{(-j)}(\cdot,\cdot):\mathcal{X}\times\Theta\rightarrow\mathbb{R}^{q}$ and estimate an initial $\sqrt{n}$-consistent estimator $\hat{\theta}^{(-j)}$ for $\theta^*$. Then we construct the PDC score for the $j$-th fold
\begin{gather*}
    \hat{S}^{(j)}=\E_{\mathcal{D}_{L}^{(j)}}\Big\{\bm{s}_{\hat{\theta}^{(-j)}}+\gamma\widehat{\bm{T}}_{ss}^{(j)}(\hat{\theta}^{(-j)})\big\{\hat{\bm{f}}^{(-j)}_{\hat{\theta}^{(-j)}}-\E_{\mathcal{D}_{U}^{(j)}}(\hat{\bm{f}}^{(-j)}_{\hat{\theta}^{(-j)}})\big\}\Big\},
\end{gather*}
where $\E_{\mathcal{D}_{L}^{(j)}}(g)=K/n\sum_{i\in\mathcal{D}_{L}^{(j)}}g(\bm{x}_i)$, $\E_{\mathcal{D}_{U}^{(j)}}(g)=K/N\sum_{i\in\mathcal{D}_{U}^{(j)}}g(\bm{x}_i)$ for any measurable function $g(\cdot)$ defined in $\mathcal{X}$, and
\begin{gather*}
    \widehat{\bm{T}}_{ss}^{(j)}(\theta)=\cov_{\mathcal{D}_{L}^{(j)}}(\bm{s}_{\theta},\hat{\bm{f}}^{(-j)}_{\theta})\{\cov_{\mathcal{D}_L^{(j)}}(\hat{\bm{f}}_{\theta}^{(-j)})\}^{-1},
\end{gather*}
where $\cov_{\mathcal{D}_L^{(j)}}(\cdot)$ denotes the sample covariance based on the dataset $\mathcal{D}_{L}^{(j)}$. The one-step PDC estimator for the $j$-th fold is $\hat{\theta}^{(-j)}-\widehat{\bm{H}}^{-1}\hat{S}^{(j)}$.
The semi-supervised PDC estimator is the aggregation of these estimators 
\begin{gather*}
    \hat{\theta}_{sspdc,1}=\frac{1}{K}\sum_{j=1}^K \big(\hat{\theta}^{(-j)}-\widehat{\bm{H}}^{-1}\hat{S}^{(j)}\big).
\end{gather*}
Theorem \ref{thm:cross pdc} shows that, under suitable conditions, $\hat{\theta}_{sspdc,1}$ has the same asymptotic distribution as $\hat{\theta}_{pdc,1}$.

\begin{thm}\label{thm:cross pdc}
    Assume there exists a fixed real-valued function $\bm{f}$ such that Assumptions \ref{asmp:theta is interior}-\ref{asmp:smooth score} and Assumptions \ref{asmp:cross T moment}-\ref{asmp:cross f moment} in the Appendix hold. Further more, assume that $N/(n+N)\rightarrow \eta$ for some constant $\eta\in(0,1]$.  Then, for any fixed $\gamma\in\mathbb{R}$, we have
	\begin{align}\label{equ:cspdc aymp normal}
		\sqrt{n}(\hat{\theta}_{sspdc,1}-\theta^*)\overset{D}{\rightarrow} \N_d(\bm{0}, \bm{H}^{-1}\bm{\Gamma}(\gamma)\{\bm{H}^{\top}\}^{-1}).
	\end{align}
\end{thm}

\subsection{Extension: stochastic predictive model}\label{subsec:random model}
In Subsection \ref{subsec:asymp normality}, we implicitly assume that $\bm{f}(\cdot,\theta)$ is a deterministic function. However, in practice, the black-box ML model $\mu(\cdot)$ could be a random function. We now extend Theorem \ref{thm:one-step pdc} to accommodate this scenario. Here, we rewrite $\mu(\cdot)$ as $\mu_{\zeta}(\cdot)$, $\bm{f}(\cdot,\theta)$ as $\bm{f}_{\zeta}(\cdot,\theta)$, where $\zeta$ represents the stochastic noise. We need the following conditions.

\begin{asmp}\label{asmp:stochastic asmp}

 \begin{enumerate}[(a)]
     \item \label{asmp:f decompose} For any $\theta\in\Theta$, $\bm{f}_{\zeta}(\cdot,\theta)$ can be decomposed as $\bm{f}_{\zeta}(\cdot,\theta)=\tilde{\bm{f}}(\cdot,\theta)+\bm{b}_{\zeta}(\cdot)$ for some deterministic function $\tilde{\bm{f}}(\cdot,\cdot):\mathbb{R}^{p}\times\mathbb{R}^d\rightarrow\mathbb{R}^q$ and stochastic function $\bm{b}_{\zeta}(\cdot):\mathbb{R}^p\rightarrow\mathbb{R}^q$.
     \item \label{asmp:smooth s, tilde f} $\{\bm{s}(\cdot,\cdot,\theta):\theta\in\Theta\}$ is a $\Pr_{\text{Y},\mathbf{X}}$-Donsker class, \label{asmp:smooth stochastic f} $\{\tilde{\bm{f}}(\cdot,\theta):\theta\in\Theta\}$ is a $\Pr_{\mathbf{X}}$-Donsker class.
 \end{enumerate}
	
\end{asmp}
Assumption \ref{asmp:stochastic asmp}(\ref{asmp:f decompose}) is a structural assumption which ensures that the stochasticity of $\zeta$ does not affect the estimation of $\theta$.  In practice, if we choose $\bm{f}_{\zeta}(\bm{x},\theta)$ as $\bm{s}(\mu_{\zeta}(\bm{x}),\bm{x},\theta)$  for some random function $\mu_{\zeta}(\cdot)$, then Assumption \ref{asmp:stochastic asmp}(\ref{asmp:f decompose}) is satisfied by most common estimating functions, e.g., the estimating functions in Examples \ref{exmp:mean estimation}, \ref{exmp:glm}. Assumption \ref{asmp:stochastic asmp}(\ref{asmp:smooth s, tilde f}) is a generalized version of Assumption \ref{asmp:smooth score}.

Next, we extend Theorem \ref{thm:one-step pdc} to accommodate the random predictive model. Recall that $\E(\cdot)$ is defined as the expectation with respect to all the stochastic terms, so now the expectations in Theorem \ref{thm:random f} are taken with respect to $(\text{Y},\mathbf{X},\zeta)$.
\begin{thm}\label{thm:random f}
Suppose that Assumptions \ref{asmp:theta is interior}-\ref{asmp:covf and H is nonsingular} and \ref{asmp:stochastic asmp} hold, and $N/(n+N)\rightarrow \eta$ for some constant $\eta\in(0,1]$.   Then for any fixed $\gamma\in\mathbb{R}$, \eqref{equ:pdc aymp normal} still holds with $\bm{\Gamma}(\gamma)=\cov\{\bm{s}_{\theta^*}\}+(\gamma^2/\eta+2\gamma)\cov(\bm{s}_{\theta^*},\bm{f}_{\zeta;\theta^*})\{\cov(\bm{f}_{\zeta;\theta^*})\}^{-1}\cov(\bm{f}_{\zeta;\theta^*},\bm{s}_{\theta^*})$, where $\bm{f}_{\zeta;\theta}=\bm{f}_{\zeta}(\cdot,\theta)$.
\end{thm}
 Now we can use $\hat{\theta}_{pdc,1}$ to construct confidence interval for $\bm{c}^{\top}\theta^*$, where $\bm{c}$ is any vector from $\mathbb{R}^d$. Under the conditions of Theorem \ref{thm:one-step pdc}, the asymptotic $(1-\alpha)$-size  confidence interval for $\bm{c}^{\top}\theta^*$ can be written as 
	\begin{equation*}
			\Big[\bm{c}^{\top}\hat{\theta}_{pdc,1}-\frac{z_{1-\alpha/2}}{\sqrt{n}}\Big\{\bm{c}^{\top}\widehat{\bm{V}}(\gamma)\bm{c}\Big\}^{1/2},\bm{c}^{\top}\hat{\theta}_{pdc,1}+\frac{z_{1-\alpha/2}}{\sqrt{n}}\Big\{\bm{c}^{\top}\widehat{\bm{V}}(\gamma)\bm{c}\Big\}^{1/2}\Big],
	\end{equation*}
where $\widehat{\bm{V}}(\gamma)=\widehat{\bm{H}}^{-1}\widehat{\bm{\Gamma}}(\gamma)\{\widehat{\bm{H}}^{\top}\}^{-1}$, and $\widehat{\bm{H}}$, $\widehat{\bm{\Gamma}}(\gamma)$ are consistent estimators for $\bm{H}$, $\bm{\Gamma}(\gamma)$ respectively.

For consistent estimators for $\bm{H}$ and $\bm{\Gamma}(\gamma)$, we can choose
\begin{gather*}
	\widehat{\bm{H}}=\E_n\{\nabla_{\theta^{\top}}\bm{s}(\text{Y},\mathbf{X},\hat{\theta})\},\\
	\widehat{\bm{\Gamma}}(\gamma)=\cov_n\{\bm{s}_{\hat{\theta}}\}+(\frac{1}{\eta}\gamma^2+2\gamma)\cov_n(\bm{s}_{\hat{\theta}},\bm{f}_{\hat{\theta}})\{\cov_n(\bm{f}_{\hat{\theta}})\}^{-1}\cov_n(\bm{f}_{\hat{\theta}},\bm{s}_{\hat{\theta}}),
\end{gather*}
where $\hat{\theta}$  is a consistent estimator of $\theta^*$ that can be replaced by any appropriate estimator of $\theta^*$, for example, the aforementioned one-step PDC estimator $\hat{\theta}_{pdc,1}$. To fairly compare different estimators, we choose $\hat{\theta}=\hat{\theta}_{sup}$ since it only utilizes the labeled data. Then the asymptotic $(1-\alpha)$-size  confidence interval for $\bm{c}^{\top}\theta^*$ constructed by supervised estimator $\hat{\theta}_{sup}$ can be written as
\begin{equation*}
	\Big[\bm{c}^{\top}\hat{\theta}_{sup}-\frac{z_{1-\alpha/2}}{\sqrt{n}}\{\bm{c}^{\top}\widehat{\bm{V}}(0)\bm{c}\}^{1/2},\bm{c}^{\top}\hat{\theta}_{sup}+\frac{z_{1-\alpha/2}}{\sqrt{n}}\{\bm{c}^{\top}\widehat{\bm{V}}(0)\bm{c}\}^{1/2}\Big].
\end{equation*}

\section{Experiments}\label{sec:experiments}
\subsection{Simulation studies}\label{subsec:simulation}
In this section, we investigate the numerical performance of PDC estimator in various settings in terms of coverage probability and width of confidence intervals (CI). We compare the performance of PDC estimator to that of the following estimators: PPI++, PPI \citep{angelopoulos2023prediction}, semi-supervised estimator \citep{Zhang2019ssmean,azriel2022semi}, and supervised estimator $\hat{\theta}_{sup}$. Since the methods DESSL \citep{schmutz2023dont} and PPI++ \citep{angelopoulos2023ppi++} are almost equivalent in our simulation settings, we only report the results of PPI++. For the PDC, PPI++, PPI, we use the one-step version. Confidence levels are all set to $\alpha=0.1$. All simulation results are based on 1000 replications. We generate data from model  $\text{Y}=1+\bm{\beta}^\top\mathbf{X}+\bm{\beta}^\top(\mathbf{X}^2-\bm{1})+\zeta$ with  $\mathbf{X}\sim \N(0,\bm{I})$, $\zeta\sim \N(0,1)$ and $\bm{\beta}^{\top}=(\bm{\beta}_{(1)},-1,-2,-2)$, The value of $\bm{\beta}_{(1)}$ may be different under different settings. Unless specified, the sample size for the labeled data is $n=1000$, while for the unlabeled data it is $N=5000$. Predictive estimating functions $\bm{f}(\bm{x},\theta)$ are set to $\bm{s}(\mu(\bm{x}),\bm{x},\theta)$ with some predictive model $\mu(\cdot)$.

We report simulation results in terms of the coverage probability of CI and the ratio of CI's width to that of the supervised counterpart. The reported width ratio (WR) is the average of width ratios from 1000 replications. We consider the following two scenarios:
\begin{enumerate}[\text{Scenario} 1]
	\item (Mean estimation) The estimation target is $\theta^*=\E(\text{Y})$. $\bm{s}(y,\bm{x},\theta)=\theta-y$.  Since $\theta^*$ is a 1-dimensional target and $\bm{f}(\bm{x},\theta)=\bm{s}(\mu(\bm{x}),\bm{x},\theta)$, PDC is equivalent to PPI++ in this setup. We compare the performance of PDC estimator to the performance of three estimators: semi-supervised estimator \citep{Zhang2019ssmean}, PPI, and supervised estimator. 
	\item (Linear regression) The estimation target is $\theta^*_{(1)}$, where $\theta^*=\arg\min_{\theta}\E\{(\text{Y}-\theta^\top\mathbf{X})^2\}$. $\bm{s}(y,\bm{x},\theta)=(\bm{x}^\top\theta-y)\bm{x}$. We compare the performance of PDC to the performance of four existing methods: semi-supervised estimator \citep{azriel2022semi}, PPI++, PPI, and supervised estimator. 
\end{enumerate}

\subsubsection{Mean estimation}
In this scenario, we consider three different settings, which take into account the influence of the magnitude of stochastic noise, labeled sample size, and unlabeled sample size on the performance of different methods.   For each setting, we fix $\bm{\beta}_{(1)}=9$, and generate $\zeta'$ for each observation.  
\begin{enumerate}[\text{Setting} 1]
    \item \label{setting 1} The predictive model is set as $\mu_{\zeta'}(\bm{x})=(1+\epsilon\cdot\zeta')(\bm{\beta}^\top\bm{x}^2)$, where $\zeta'\sim \N(0,1)$, $\epsilon$ is successively set to each value in the set $\{0,0.2, 0.4,\dots,2\}$. 
    \item \label{setting 2}The predictive model is set as $\mu_{\zeta'}(\bm{x})=(1+0.5\cdot\zeta')(\bm{\beta}^\top\bm{x}^2)$, where $\zeta'\sim \N(0,1)$. $n$ is successively set to each value in the set ${500,1000,\dots,5500}$, $N$ is fixed at 5000. 
    \item \label{setting 3}The predictive model is set as $\mu_{\zeta'}(\bm{x})=(1+0.5\cdot\zeta')(\bm{\beta}^\top\bm{x}^2)$, where $\zeta'\sim \N(0,1)$. $N$ is successively set to each value in the set $\{2000,4000,\dots,22000\}$, $n$ is fixed at 1000.
\end{enumerate}
\begin{figure}[htbp]
	\centering
	
	\subfigure[Coverage probability for $\E \text{Y}$]{
		\includegraphics[width=0.45\textwidth]{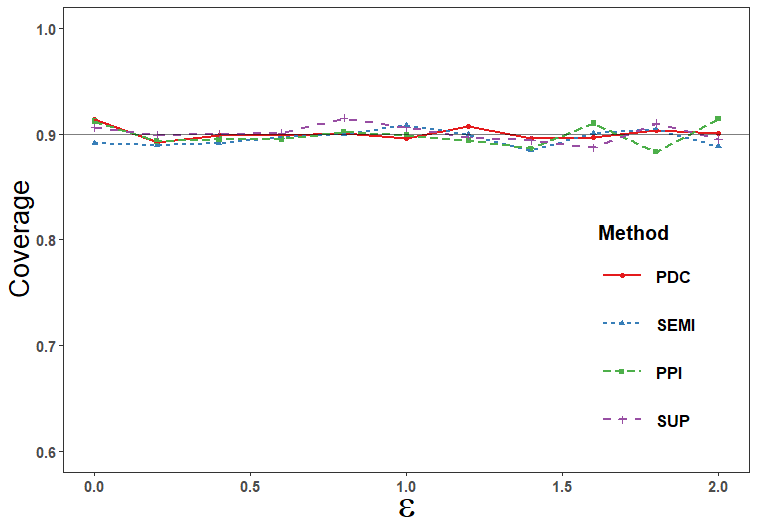}
		\label{fig:mean_cov}}\quad
	\subfigure[Width ratio]{
		\includegraphics[width=0.45\textwidth]{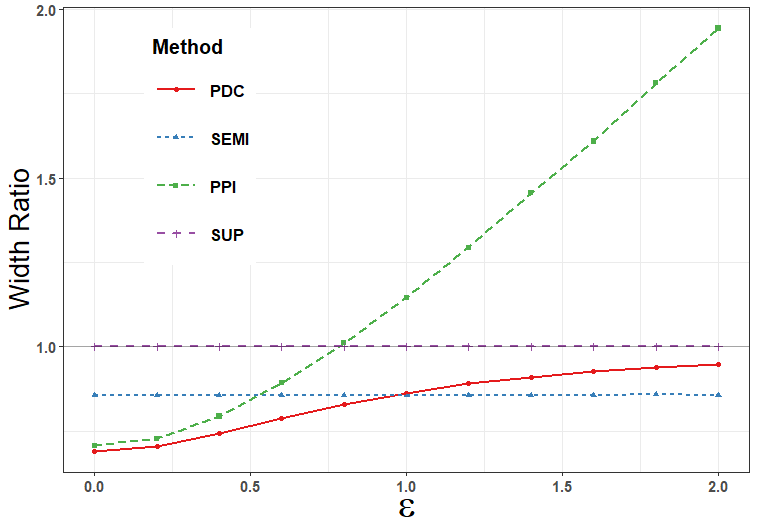}
		\label{fig:mean_wid}
	}
	\caption{\footnotesize Simulation results under Setting \ref{setting 1}.  Figure \ref{fig:mean_cov} presents the coverage probability of CI constructed by four methods versus noise level $\epsilon$ of the predictive model $\mu(\cdot)$. Figure \ref{fig:mean_wid} presents the ratio of the width of different CI  to the width of the CI constructed by $\hat{\theta}_{sup}$, versus the noise level $\epsilon$.}
	\label{fig:mean estimation vs epsilon}
\end{figure}

\begin{figure}[htbp]
	\centering
	
	\subfigure[Coverage probability for $\E \text{Y}$]{
		\includegraphics[width=0.45\textwidth]{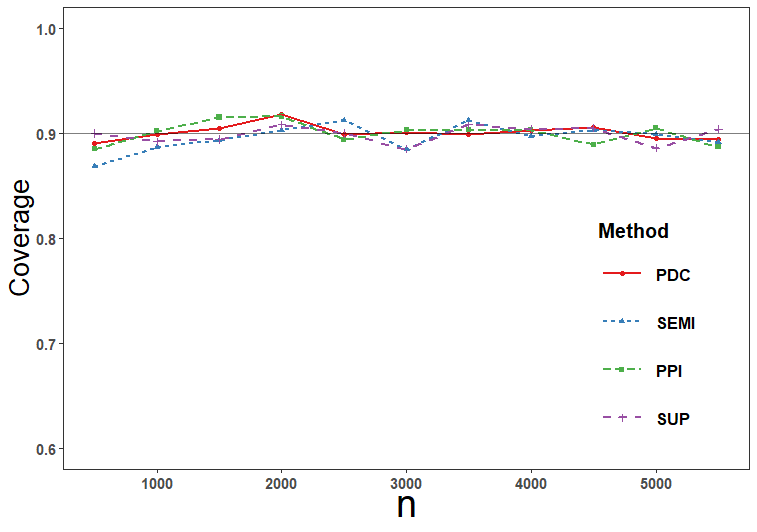}
		\label{fig:mean_cov vs n}}\quad
	\subfigure[Width ratio]{
		\includegraphics[width=0.45\textwidth]{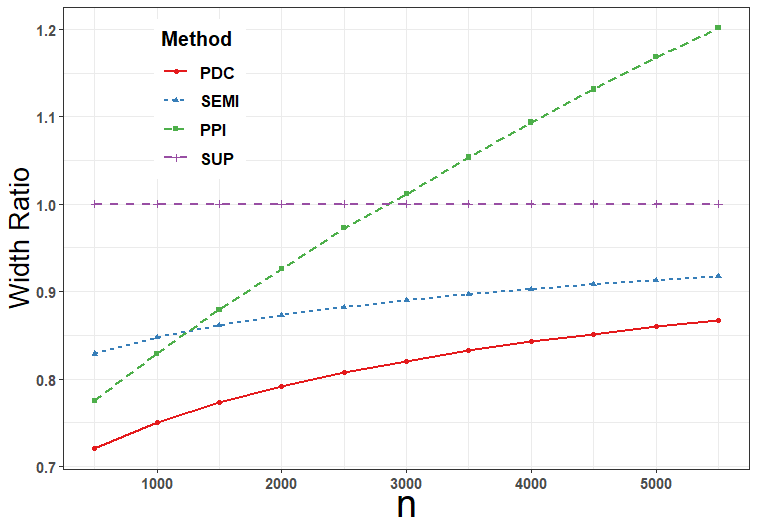}
		\label{fig:mean_wid vs n}
	}
	\caption{\footnotesize Simulation results under Setting \ref{setting 2}.  Figure \ref{fig:mean_cov vs n} presents the coverage probability of CI constructed by four methods versus labeled sample size $n$. Figure \ref{fig:mean_wid vs n} presents ratio of the width of different CI  to the width of the CI constructed by $\hat{\theta}_{sup}$, versus labeled sample size $n$.}
	\label{fig:mean estimation vs n}
\end{figure}

\begin{figure}[htbp]
	\centering
	
	\subfigure[Coverage probability for $\E \text{Y}$]{
		\includegraphics[width=0.45\textwidth]{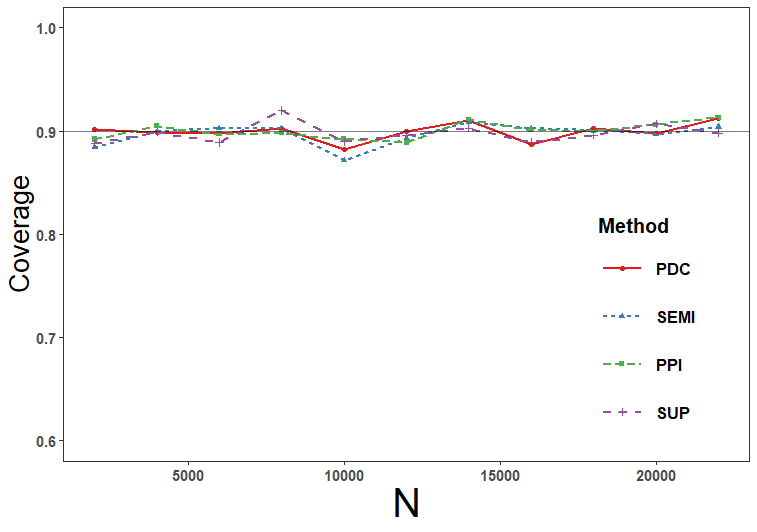}
		\label{fig:mean_cov vs N}}\quad
	\subfigure[Width ratio]{
		\includegraphics[width=0.45\textwidth]{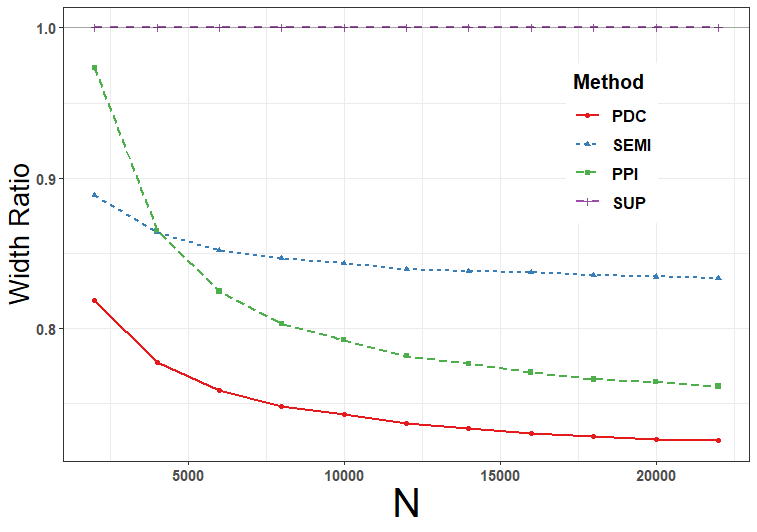}
		\label{fig:mean_wid vs N}
	}
	\caption{\footnotesize Simulation results under Setting \ref{setting 3}.  Figure \ref{fig:mean_cov vs N} presents the coverage probability of CI constructed by four methods versus unlabeled sample size $N$. Figure \ref{fig:mean_wid vs N} presents ratio of the width of different CI  to the width of the CI constructed by $\hat{\theta}_{sup}$, versus the unlabelled sample size $N$.}
	\label{fig:mean estimation vs N}
\end{figure}

Figures \ref{fig:mean estimation vs epsilon}-\ref{fig:mean estimation vs N} summarizes the simulation results under  Settings \ref{setting 1}-\ref{setting 3}, respectively.   Figures \ref{fig:mean_cov}, \ref{fig:mean_cov vs n} and \ref{fig:mean_cov vs N}  show that all methods guarantee the nominal $(1-\alpha)$ coverage rate. For the width, Figure \ref{fig:mean_wid} demonstrates that the performance of PPI  deteriorates rapidly with the increase of noise level and can be substantially worse than the supervised counterpart. In contrast, the performance of PDC gets worse slowly and it consistently remains superior to the supervised counterpart. Both the PPI and PDC outperform the semi-supervised method when the noise level is low, whereas both methods perform worse than the semi-supervised method when the noise level is high. Figure \ref{fig:mean_wid vs n} demonstrates that PPI outperforms the supervised method when $n$ is relatively small compared to $N$. However, as $n$ increases to the same order as $N$, the PPI's performance deteriorates and becomes inferior to the supervised method. Conversely, PDC consistently maintains better performance than the supervised method, regardless of the value of $n$.  Figure \ref{fig:mean_wid vs N} further supports the results of Figure \ref{fig:mean_wid vs n}. When $N$ is of the same order as $n$, the performance of PPI is inferior to that of the semi-supervised method. However, as $N$ gradually increases to be much larger than $n$, the PPI quickly outperforms the semi-supervised method. The PDC consistently outperforms other methods regardless of the value of $N$. 

\subsubsection{Linear regression}
We consider the following three settings in the linear regression scenario. In setting \ref{setting 4} and setting \ref{setting 5}, we investigate the influence of the coefficient $\bm{\beta}_{(1)}$ and use two different predictive models, corresponding to poor and good models, respectively. In setting \ref{setting 6}, we explore the performance of the PDC procedure after combining two models as (\ref{equ:combine}). In all settings, $\bm{\beta}_{(1)}$ takes value from the set $\{0,1,\dots,10\}$ successively.

\begin{enumerate}[\text{Setting} 1]
\setcounter{enumi}{3} 
    \item \label{setting 4}The predictive model is set as $\mu(\bm{x})=\bm{\beta}_{-(1)}^\top\bm{x}_{-(1)}+\bm{\beta}_{-(1)}^\top\bm{x}_{-(1)}^2$. 
    \item \label{setting 5}The predictive model is set as $\mu(\bm{x})=\bm{\beta}^\top\bm{x}^2$.
    \item \label{setting 6}Two random forest models \citep{breiman2001random} are trained from two datasets with distributions different from the target population. The two datasets are generated as follows:
    \begin{enumerate}[\text{Dataset} 1]
        \item $\mathbf{X}\sim \N(0,\bm{I})$, $\text{Y}=\bm{\beta}^\top\mathbf{X}+\bm{\beta}^\top\mathbf{X}^3+\bm{\beta}^\top\exp{(\mathbf{X})}+\zeta_1$;
        \item $\mathbf{X}\sim \N(0,\bm{I})$, $\text{Y}=\bm{\beta}^\top\mathbf{X}^2+\zeta_2$,
    \end{enumerate}
    where $\zeta_1\sim\text{LN}(0,1)$, $\zeta_2\sim t_3$, $\text{LN}$ and $t$ represent the log-normal distribution and  $t$ distribution respectively.
\end{enumerate}

\begin{figure}[htbp]
	\centering
	
	\subfigure[Coverage probability for $\bm{\theta}_{(1)}$]{
		\includegraphics[width=0.45\textwidth]{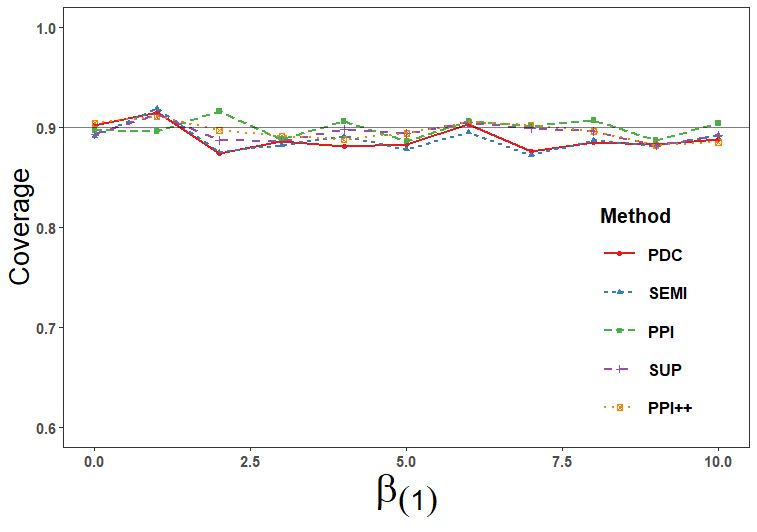}
		\label{fig:linear cov bad}}\quad
	\subfigure[Width ratio]{
		\includegraphics[width=0.45\textwidth]{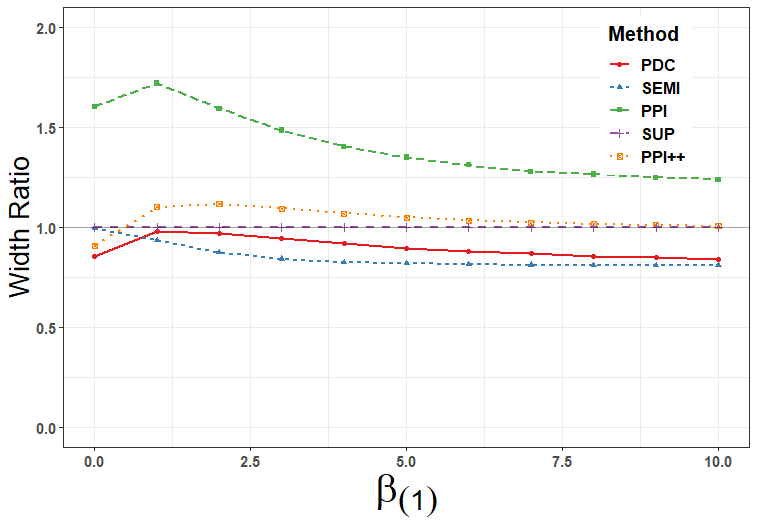}
		\label{fig:linear wid bad}
	}
	\caption{\footnotesize Simulation results under Setting \ref{setting 4}.  Figure \ref{fig:linear cov bad} presents the coverage probability of CI constructed using various methods versus the coefficient $\bm{\beta}_{(1)}$. Figure \ref{fig:linear wid bad} presents the ratio of the widths of different CI constructed using various methods to the width of the CI constructed by $\hat{\theta}_{sup}$, versus the coefficient $\bm{\beta}_{(1)}$.}
	\label{fig:linear bad model}
\end{figure}

\begin{figure}[htbp]
	\centering
	
	\subfigure[Coverage probability for $\bm{\theta}_{(1)}$]{
		\includegraphics[width=0.45\textwidth]{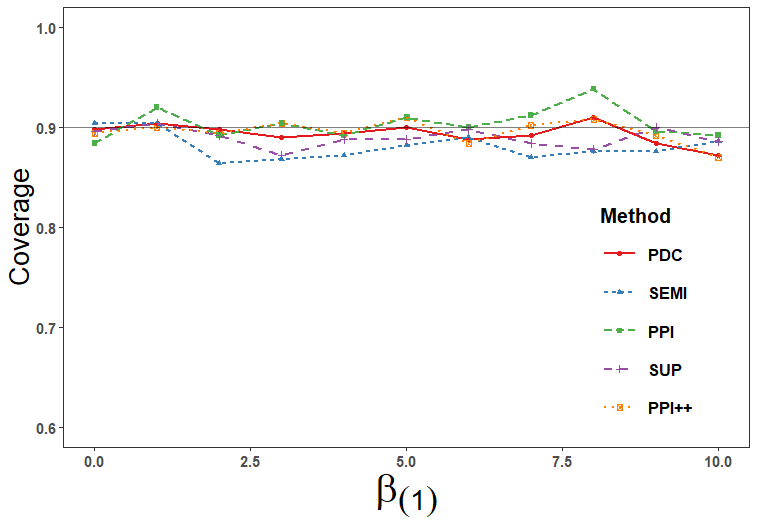}
		\label{fig:linear cov good}}\quad
	\subfigure[Width ratio]{
		\includegraphics[width=0.45\textwidth]{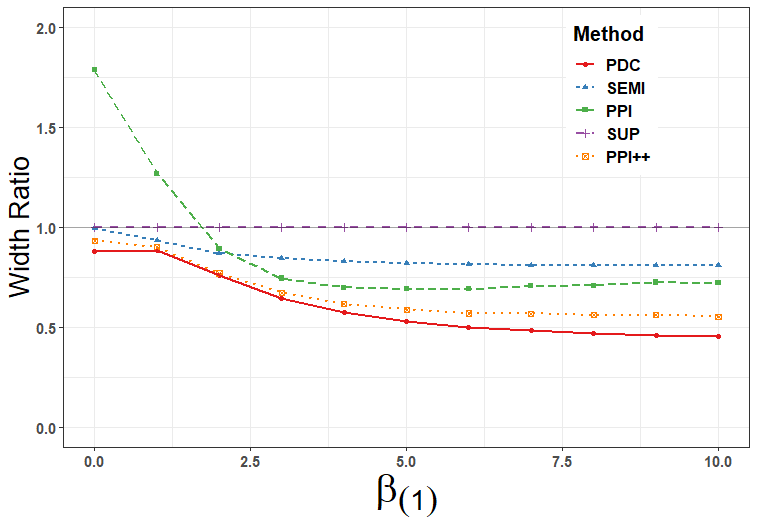}
		\label{fig:linear wid good}
	}
	\caption{\footnotesize Simulation results under Setting \ref{setting 5}.  Figure \ref{fig:linear cov good} presents the coverage probability of CI constructed using various methods versus the coefficient $\bm{\beta}_{(1)}$. Figure \ref{fig:linear wid good} presents the ratio of the widths of different CI constructed using various methods to the width of the CI constructed by $\hat{\theta}_{sup}$, versus the coefficient $\bm{\beta}_{(1)}$.}
	\label{fig:linear good model}
\end{figure}

\begin{figure}[htbp]
	\centering
	
	\subfigure[Coverage probability for $\bm{\theta}_{(1)}$]{
		\includegraphics[width=0.45\textwidth]{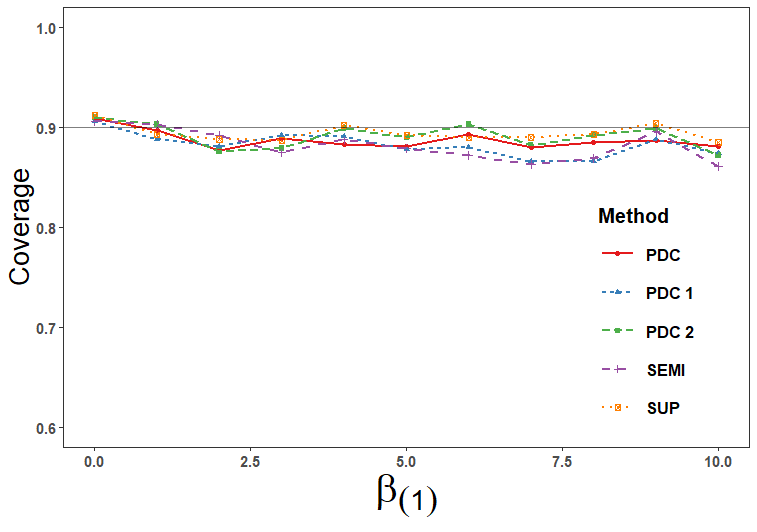}
		\label{fig:two model cov}}\quad
	\subfigure[Width ratio]{
		\includegraphics[width=0.45\textwidth]{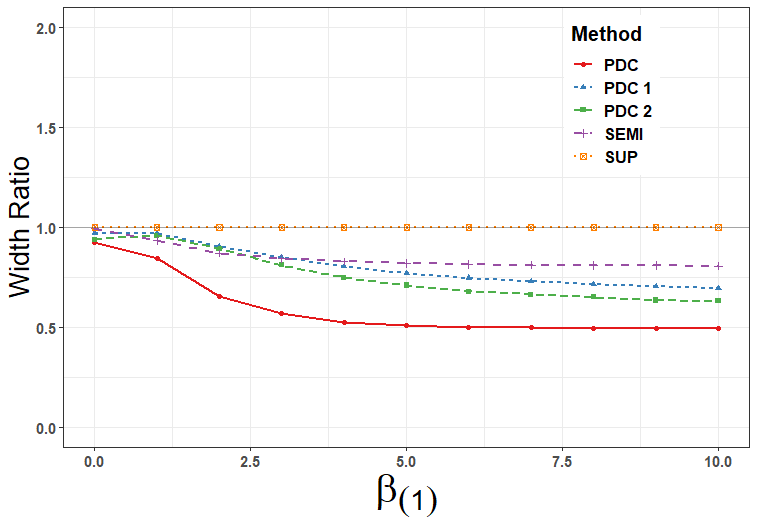}
		\label{fig:two model wid}
	}
	\caption{\footnotesize Simulation results under Setting \ref{setting 6}.  Figure \ref{fig:two model cov} presents the coverage probability of CI constructed using various methods versus the coefficient $\bm{\beta}_{(1)}$. Figure \ref{fig:two model wid} presents the ratio of the widths of different CI constructed using various methods to the width of the CI constructed by $\hat{\theta}_{sup}$, versus the coefficient $\bm{\beta}_{(1)}$. `PDC 1' and `PDC 2' refer to the PDC method utilizing random forest models trained from Dataset 1 and Dataset 2, respectively. `PDC' refers to the PDC method utilizing two random forest models.}
	\label{fig:two model pdc}
\end{figure}

Figures \ref{fig:linear bad model}-\ref{fig:two model pdc} present the simulation results under the Settings \ref{setting 4}-\ref{setting 6}.  Figures \ref{fig:linear cov bad}, \ref{fig:linear cov good}  and \ref{fig:two model cov} show that all methods guarantee the nominal $(1-\alpha)$ coverage rate. When the predictive model is relatively poor, Figure \ref{fig:linear wid bad}  shows that the PPI method is worse than the supervised method for all values of $\bm{\beta}_{(1)}$. The PPI++ method outperforms the supervised method only when $\bm{\beta}_{(1)}$ is very small, which is also the case where the predictive model performs relatively well. However, as $\bm{\beta}_{(1)}$ increases, PPI++ quickly becomes inferior to the supervised method.  Both the semi-supervised method and PDC outperform the supervised method.  Figure \ref{fig:linear wid bad}  also implies that when the predictive model is poor, both the PPI and  PPI++, as well as the PDC, can be inferior to the semi-supervised method.  When the predictive model is relatively good, Figure \ref{fig:linear wid good}  shows that all the methods can outperform the supervised counterpart when $\bm{\beta}_{(1)}$ is relatively large, and particularly, the PDC outperforms all the other methods regardless of the value of $\bm{\beta}_{(1)}$. Figure \ref{fig:two model wid} shows that combining two predictive models can substantially  reduce the variance of the PDC estimator, and consequently, narrow the width of the CI constructed by the PDC.

\subsection{Real data analysis}\label{subsec:real data}
In this section, we consider an application of our proposed method to the Los Angeles homeless dataset \citep{kriegler2010small}. In 2004-2005, the Los Angeles Homeless Services Authority conducted a study of the homeless population. A stratified spatial sampling was used on the 2054 census tracts of Los Angeles County. First, 244 tracts that were believed to have a  large amount of homeless were preselected and visited with probability 1. 
These tracts are known as "hot tracts." Next, 265 of the rest non-hot tracts were randomly selected and visited (as $\mathcal{D}_L$), leaving 1545 non-hot tracts unvisited (as $\mathcal{D}_U$). We include seven predictors recommended by \cite{kriegler2010small} in our study, \texttt{Perc.Industrial}, \texttt{Perc.Residential}, \texttt{Perc.Vacant}, \texttt{Perc.Commercial}, \texttt{Perc.OwnerOcc}, \texttt{Perc.Minority}, \texttt{MedianHouseholdIncome}. The response variable is \texttt{Total}. Our goal is to estimate the linear regression coefficient on the population of non-hot tracts. 

Although the data from 244 hot tracts is labeled, it is considered to be biased from the population of non-hot tracts. The previous works usually leave out this part of data and only utilize the data from 1810 non-hot tracts \citep{Zhang2019ssmean,azriel2022semi,song2023general}. To fully leverage the information contained in the dataset, we train a random forest model \citep{breiman2001random} on the data from 244 hot tracts and use this model as the predictive function  $\mu(\cdot)$. The resulting estimators,  along with their estimated standard deviation ($\widehat{\text{SD}}$) and  the ratio of their CI's width to that of the supervised counterpart are given in Table \ref{tab:homeless}. 

Since the predictive model is obtained from a different distribution, the performance of the PPI estimators is poor with much larger standard errors than the supervised counterparts. We owe this to the prediction bias. Both the PDC and PPI++ produce reasonable estimates. We also notice that the PDC estimator always has a smaller standard error than that of the other three methods. Another point worth noting is that, for the estimation of the coefficient of \texttt{Perc.Industrial}, PPI++ produces an estimator with a larger standard error than that of the supervised counterpart. This demonstrates that the PPI++ cannot guarantee the element-wise variance reduction. To see the difference between PPI++ and PDC more fundamentally, we present the numerical results of the estimated $\bm{T}$ and the tuning parameter of PPI++ at the end of the Appendix.

\begin{table}
	\centering
	\caption{\footnotesize Results from applying different methods to the Los Angeles homeless dataset described in Subsection \ref{subsec:real data}. }
 \small
    \tabcolsep 1.5pt
    \begin{threeparttable}
	\begin{tabular}{l|cc|ccc|ccc|ccc}
		\hline
		~&\multicolumn{2}{c}{SUP}&\multicolumn{3}{c}{PDC}&\multicolumn{3}{c}{PPI++}&\multicolumn{3}{c}{PPI}\\
        \hline
        ~&$\hat{\theta}_{sup}$&$\widehat{\text{SD}}$&$\hat{\theta}_{pdc}$&$\widehat{\text{SD}}$&WR&$\hat{\theta}_{ppi++}$&$\widehat{\text{SD}}$&WR&$\hat{\theta}_{ppi}$&$\widehat{\text{SD}}$&WR\\
		\hline
		Intercept&21.964&14.171&16.906&13.633&0.962&20.603&14.023&0.990&48.419& 41.258&2.911\\
		Perc.Industrial&0.028&0.194&-0.018&0.188&0.966&0.032&0.203&1.044&-0.049&  0.621&3.197\\
		Perc.Residential&-0.088&0.097&-0.108&0.096&0.983&-0.098&0.096&0.991&0.115&  0.308&3.169\\
		Perc.Vacant&1.404&0.570&1.995&0.474&0.832&1.478&0.541&0.948&-0.033&1.932&  3.387\\
		Perc.Commercial&0.339&0.362&0.515&0.350&0.965&0.360&0.356&0.983&-0.071&  0.783&2.160\\
		Perc.OwnerOcc&-0.233&0.090&-0.212&0.085&0.940&-0.209&0.088&0.980&-0.703&  0.255&2.830\\
		Perc.Minority&0.059&0.130&0.077&0.125&0.956&0.071&0.129&0.990&-0.186&0.327&  2.504\\
		MedianInc (in \$K)&0.074&0.081&0.079&0.069&0.846&0.066&0.073&0.901&  0.235&0.415&5.104\\
		\hline
	\end{tabular}
 \begin{tablenotes}
     \footnotesize
     \item \textit{Notes}: $\hat{\theta}$, the point estimator; $\widehat{\text{SD}}$, the estimated standard deviation; WR, the ratio of CI‘s width to that of the supervised counterpart.
 \end{tablenotes}
	\end{threeparttable}
	\label{tab:homeless}
\end{table}

\section{Discussion}\label{sec:discussion}
In this article, we propose a novel assumption-lean, model-adaptive approach for post-prediction inference, called Prediction De-Correlated inference. We establish the asymptotic normality of the PDC estimator and compare it with several recent works. The PDC estimator consistently outperforms the supervised estimator  and can be extended to accommodate multiple predictive models. Both numerical results and real-world data analysis back up our theoretical results. 

This article inspires several directions worthy of further investigation. First, the optimal approach to combine multiple predictive models is unclear. Second, the PDC procedure can only be applied when the dimension of inference target is fixed. Finding a safe alternative under high-dimensional post-prediction settings is of interest and we leave it for future work.

\section*{Acknowledgements}
The authors are indebted to the Editor, the Associate Editor and two referees for their constructive comments and suggestions that led to a significant improvement of this paper.

\bibliography{ref.bib}
\bibliographystyle{asa}

\appendix\label{sec:appendix}
\section*{Appendix}

We define the empirical notation for any measurable  function $g(\cdot)$, $\mathbb{G}_n (g)=\sqrt{n}\{\E_n (g)-\E (g)\}$, $\mathbb{G}_N (g)=\sqrt{N}\{\E_N (g)-\E (g)\}$.

\begin{proof}[Proof of Theorem \ref{thm:one-step pdc}]
    \begin{equation}
		\begin{split}
			\sqrt{n}\widehat{\bm{H}}(\hat{\theta}_{pdc,1}-\theta^*)&=\sqrt{n}\widehat{\bm{H}}(\hat{\theta}_0-\theta^*)-\sqrt{n}\E_n(\bm{s}_{\hat{\theta}_0}-\bm{s}_{\theta^*})\\
			&\quad-\sqrt{n}\Big[\E_n\bm{s}_{\theta^*}+\gamma\bm{T}^*(\theta^*)\{\E_n (\bm{f}_{\theta^*})-\E_N (\bm{f}_{\theta^*})\}\Big]\\
			&\quad-\sqrt{n}\gamma\bm{T}^*(\theta^*)\big\{\E_n (\bm{f}_{\hat{\theta}_0}-\bm{f}_{\theta^*})-\E_N(\bm{f}_{\hat{\theta}_0}-\bm{f}_{\theta^*})\big\}\\
			&\quad-\sqrt{n}\gamma\{\widehat{\bm{T}}(\hat{\theta}_0)-\bm{T}^*(\theta^*)\}\{\E_n (\bm{f}_{\hat{\theta}_0})-\E_N (\bm{f}_{\hat{\theta}_0})\}\\
			&:=\RN{1}-\RN{2}-\RN{3}-\RN{4}.
		\end{split}\label{equ:one-step decompose}
	\end{equation}
	
	By Assumption \ref{asmp:smooth score} and Example 19.7 in \cite{vaart_1998}, we know function class $\{\bm{s}(\cdot,\cdot,\theta):\theta\in\Theta\}$ is a $\Pr_{\text{Y},\mathbf{X}}$-Donsker class.  Next, Lemma 19.24 of \cite{vaart_1998} implies
	\begin{gather}\label{equ:by lem 19.24}
		\mathbb{G}_n(\bm{s}_{\hat{\theta}_0}-\bm{s}_{\theta^*})=o_p(1).
	\end{gather}
	Hence, 
	\begin{align}\label{equ:E_n s_0-s^*}
		\E_n(\bm{s}_{\hat{\theta}_{0}}-\bm{s}_{\theta^*})=\E(\bm{s}_{\hat{\theta}_{0}}-\bm{s}_{\theta^*})+o_p(n^{-1/2}).
	\end{align}
	
	Thus we can write $\RN{1}$ as
	\begin{equation}
		\begin{split}
			\RN{1}&=\sqrt{n}\widehat{\bm{H}}(\hat{\theta}_0-\theta^*)-\sqrt{n}\E(\bm{s}_{\hat{\theta}_{0}}-\bm{s}_{\theta^*})+o_p(1)\\
			&=\sqrt{n}\widehat{\bm{H}}(\hat{\theta}_0-\theta^*)-\sqrt{n}\bm{H}(\hat{\theta}_0-\theta^*)+\sqrt{n}o_p(\Vert\hat{\theta}_0-\theta^*\Vert)+o_p(1)\\
			&=o_p(1),
		\end{split}\label{equ:I=o_p(1)}
	\end{equation}
	where the second line uses Taylor expansion, the last line follows from Assumption \ref{asmp:consist init}.
	
	For $\RN{3}$, following the same step as deriving \eqref{equ:by lem 19.24}, we have
 \begin{small}
	\begin{equation}
		\begin{split}
			\E_n(\bm{f}_{\hat{\theta}_{0}}-\bm{f}_{\theta^*})=\E(\bm{f}_{\hat{\theta}_{0}}-\bm{f}_{\theta^*})+n^{-1/2}\mathbb{G}_n(\bm{f}_{\hat{\theta}_{0}}-\bm{f}_{\theta^*})=\E(\bm{f}_{\hat{\theta}_{0}}-\bm{f}_{\theta^*})+o_p(n^{-1/2}),\\
			\E_N(\bm{f}_{\hat{\theta}_{0}}-\bm{f}_{\theta^*})=\E(\bm{f}_{\hat{\theta}_{0}}-\bm{f}_{\theta^*})+N^{-1/2}\mathbb{G}_N(\bm{f}_{\hat{\theta}_{0}}-\bm{f}_{\theta^*})=\E(\bm{f}_{\hat{\theta}_{0}}-\bm{f}_{\theta^*})+o_p(N^{-1/2}).\label{equ:E_n f_0-f^*}
		\end{split}
	\end{equation}
 \end{small}
	Therefore,
	\begin{align}\label{equ:III=o_p(1)}
		\Vert \RN{3}\Vert&\le \gamma\Vert \bm{T}^*(\theta^*) \Vert \Big\Vert\mathbb{G}_n(\bm{f}_{\hat{\theta}_0}-\bm{f}_{\theta^*})-\sqrt{\frac{n}{N}}\mathbb{G}_N(\bm{f}_{\hat{\theta}_0}-\bm{f}_{\theta^*})\Big\Vert=o_p(1).
	\end{align}
	
	For $\RN{4}$, we decompose $\Vert \RN{4}\Vert$ as two terms,
	\begin{gather}\label{equ:IV decompose}
			\Vert \RN{4} \Vert\le\sqrt{n}\gamma\Vert\widehat{\bm{T}}(\hat{\theta}_0)-\bm{T}^*(\theta^*)\Vert\Big\Vert\E_n (\bm{f}_{\hat{\theta}_0}-\bm{f}_{\theta^*})-\E_N(\bm{f}_{\hat{\theta}_0}-\bm{f}_{\theta^*})\Big\Vert\\
			\quad+\sqrt{n}\gamma\Vert\widehat{\bm{T}}(\hat{\theta}_0)-\bm{T}^*(\theta^*)\Vert\Big\Vert\E_n (f_{\theta^*})-\E_N (f_{\theta^*})\Big\Vert.\notag
	\end{gather}
	We consider $\Vert\widehat{\bm{T}}(\hat{\theta}_0)-\bm{T}^*(\theta^*)\Vert$ first. We have
	\begin{equation}\label{equ:T diff}
		\Vert\widehat{\bm{T}}(\hat{\theta}_0)-\bm{T}^*(\theta^*)\Vert\le \Vert \widehat{\bm{T}}(\hat{\theta}_0)-\hat{\bm{T}}(\theta^*)\Vert+\Vert\widehat{\bm{T}}(\theta^*)-\bm{T}^*(\theta^*)\Vert.
	\end{equation}
	Since $d$ is fixed, the law of large numbers implies the second part on the right-hand side is a small term, that is 
	\begin{equation}\label{equ:T diff 2th term}
		\Vert\widehat{\bm{T}}(\theta^*)-\bm{T}^*(\theta^*)\Vert=o_p(1).
	\end{equation}
	For the first part, 
	\begin{align}\label{equ:T diff 1st term}
		\Vert \widehat{\bm{T}}(\hat{\theta}_0)-\widehat{\bm{T}}(\theta^*)\Vert&\le \Vert\cov_n(\bm{s}_{\hat{\theta}_0},\bm{f}_{\hat{\theta}_0})[\{\cov_n(\bm{f}_{\hat{\theta}_0})\}^{-1}-\{\cov_n(\bm{f}_{\theta^*})\}^{-1}]\Vert\notag\\
		&\quad+\Vert[\cov_n(\bm{s}_{\hat{\theta}_0},\bm{f}_{\hat{\theta}_0})-\cov_n(\bm{s}_{\theta^*},\bm{f}_{\theta^*})]\{\cov_n(\bm{f}_{\theta^*})\}^{-1}\Vert.
	\end{align}
	Observe that 
	\begin{align*}
		\cov_n(\bm{s}_{\hat{\theta}_0},\bm{f}_{\hat{\theta}_0})-\cov_n(\bm{s}_{\theta^*},\bm{f}_{\theta^*})&=\E_n(\bm{s}_{\hat{\theta}_0}\bm{f}_{\hat{\theta}_0}^{\top})-\E_n(\bm{s}_{\theta^*}\bm{f}_{\theta^*}^{\top})\\
		&\quad-[\E_n(\bm{s}_{\hat{\theta}_0})\{\E_n(\bm{f}_{\hat{\theta}_0})\}^{\top}-\E_n(\bm{s}_{\theta^*})\{\E_n(\bm{f}_{\theta^*})\}^{\top}]\\
		&=\bm{C}_1-\bm{C}_2,
	\end{align*}
	For $\bm{C}_1$, we decompose it as
	\begin{align}\label{equ:C_1 decompose}
		\bm{C}_1&=\E_n\{(\bm{s}_{\hat{\theta}_0}-\bm{s}_{\theta^*})\bm{f}_{\theta^*}^{\top}\}+\E_n\{(\bm{s}_{\hat{\theta}_0}-\bm{s}_{\theta^*})(\bm{f}_{\hat{\theta}_0}-\bm{f}_{\theta^*})^{\top}\}+\E_n\{\bm{s}_{\theta^*}(\bm{f}_{\hat{\theta}_0}-\bm{f}_{\theta^*})^{\top}\}\notag\\
		&=\bm{C}_{11}+\bm{C}_{12}+\bm{C}_{13}.
	\end{align}
	By Assumptions \ref{asmp:score is L2 bounded}-\ref{asmp:smooth score}, we have  
	\begin{gather}\label{equ:C_11 and C_13}
		\Vert\bm{C}_{11}\Vert\le\E_n(\Vert\bm{s}_{\hat{\theta}_0}-\bm{s}_{\theta^*}\Vert\cdot\Vert\bm{f}_{\theta^*}\Vert)\le \Vert\hat{\theta}_0-\theta^*\Vert\E_n(\dot{s}\cdot\Vert\bm{f}_{\theta^*}\Vert)=O_p(n^{-1/2}),\\
		\Vert\bm{C}_{13}\Vert\le\E_n(\Vert\bm{f}_{\hat{\theta}_0}-\bm{f}_{\theta^*}\Vert\cdot\Vert\bm{s}_{\theta^*}\Vert)\le \Vert\hat{\theta}_0-\theta^*\Vert\E_n(\dot{f}\cdot\Vert\bm{s}_{\theta^*}\Vert)=O_p(n^{-1/2}),\\
		\Vert\bm{C}_{12}\vert\le\E_n(\Vert\bm{s}_{\hat{\theta}_0}-\bm{s}_{\theta^*}\Vert\Vert\bm{f}_{\hat{\theta}_0}-\bm{f}_{\theta^*}\Vert)\le\E_n(\dot{s}\cdot\dot{f})\Vert\hat{\theta}_0-\theta^*\Vert^2=O_p(n^{-1}).
	\end{gather}
	On the other hand, \eqref{equ:E_n s_0-s^*} and \eqref{equ:E_n f_0-f^*}  together imply that
	\begin{align}\label{equ:C_2}
		\bm{C}_{2}=\{\E_n(\bm{s}_{\theta^*})+o_p(1)\}\{\E(\bm{f}_{\theta^*})+o_p(1)\}^{\top}-\E_n(\bm{s}_{\theta^*})\{\E_n(\bm{f}_{\theta^*})\}^{\top}=o_p(1).
	\end{align}
	Combining \eqref{equ:C_1 decompose}-\eqref{equ:C_2} yields 
	\begin{equation}\label{equ:covn0-covn^*}
		\cov_n(\bm{s}_{\hat{\theta}_0},\bm{f}_{\hat{\theta}_0})-\cov_n(\bm{s}_{\theta^*},\bm{f}_{\theta^*})=o_p(1).
	\end{equation}
	Since $d$ is fixed, by law of large numbers we have $\cov_n(\bm{s}_{\theta^*},\bm{f}_{\theta^*})\overset{p}{\rightarrow}\cov(\bm{s}_{\theta^*},\bm{f}_{\theta^*})$. This, combined with \eqref{equ:covn0-covn^*} implies that 
	\begin{equation}\label{equ:cov_n sf_0}
		\cov_n(\bm{s}_{\hat{\theta}_0},\bm{f}_{\hat{\theta}_0})=O_p(1).
	\end{equation}
	Following the same procedure, we can prove 
	\begin{equation}
		\begin{split}
			\cov_n(\bm{f}_{\theta^*})=\cov(\bm{f}_{\theta^*})+o_p(1),\\
			\cov_n(\bm{f}_{\hat{\theta}_0})=\cov(\bm{f}_{\theta^*})+o_p(1).\label{equ:covff}
		\end{split}
	\end{equation}
	Since the elements of the matrix inverse are continuous functions of the original matrix elements, \eqref{equ:covff} implies that
	\begin{gather}\label{equ:covff inverse}
		\{\cov_n(\bm{f}_{\hat{\theta}_0})\}^{-1}-\{\cov_n(\bm{f}_{\theta^*})\}^{-1}=o_p(1),\\ \{\cov_n(\bm{f}_{\theta^*})\}^{-1}=\{\cov(\bm{f}_{\theta^*})\}^{-1}+o_p(1).
	\end{gather}
	Hence, \eqref{equ:T diff 1st term}, \eqref{equ:covn0-covn^*},\eqref{equ:cov_n sf_0},\eqref{equ:covff inverse} together imply that
	\begin{equation}\label{equ:T diff 1st term o_p}
		\Vert \widehat{\bm{T}}(\hat{\theta}_0)-\widehat{\bm{T}}(\theta^*)\Vert=o_p(1).
	\end{equation}
	Combining \eqref{equ:T diff}, \eqref{equ:T diff 1st term o_p}, and \eqref{equ:T diff 2th term}, we obtain
	\begin{equation}\label{equ:T diff=o_p(1)}
		\Vert\widehat{\bm{T}}(\hat{\theta}_0)-\bm{T}^*(\theta^*)\Vert=o_p(1).
	\end{equation}
	For the rest term in  the decomposition of $\Vert \RN{4}\Vert$, by Assumption \ref{asmp:score is L2 bounded} and Markov's Inequality, 
	\begin{equation}\label{equ:f^* diff}
		\sqrt{n}\Vert\E_n (\bm{f}_{\theta^*})-\E_N (\bm{f}_{\theta^*})\Vert=O_p(1).
	\end{equation}
	This, combined with \eqref{equ:IV decompose}, \eqref{equ:T diff=o_p(1)}, and \eqref{equ:E_n f_0-f^*}, implies that
	\begin{equation}\label{equ:IV=o_p(1)}
		\Vert \RN{4}\Vert=o_p(1).
	\end{equation}
	
	Since $\mathcal{D}_L$ is independent of $\mathcal{D}_U$, by multivariate central limit theorem,
	\begin{align}\label{equ:II asymp normal}
		\RN{2}\overset{D}{\rightarrow}\N_d(\bm{0},\bm{\Gamma}(\gamma)).
	\end{align}
	Combining \eqref{equ:one-step decompose}, \eqref{equ:I=o_p(1)}, \eqref{equ:III=o_p(1)}, \eqref{equ:IV=o_p(1)}, and \eqref{equ:II asymp normal}, we have shown that
	\begin{gather*}
		\sqrt{n}\widehat{\bm{H}}(\hat{\theta}_{pdc,1}-\theta^*)\overset{D}{\rightarrow} \N_d(\bm{0},\bm{\Gamma}(\gamma)).
	\end{gather*}
	Finally, by  Slutsky's theorem and Assumption \ref{asmp:consist init}, we have completed the proof.
\end{proof}

Before we prove Theorem \ref{thm:cross pdc}, we define some notations and two extra assumptions. 
Define
\begin{gather*}
    \widehat{\bm{T}}^{(j)}(\theta)=\cov_{\mathcal{D}_{L}^{(j)}}(\bm{s}_{\theta},\bm{f}_{\theta})\{\cov_{\mathcal{D}_{L}^{(j)}}(\bm{f}_{\theta})\}^{-1},\\
    \bm{u}^{(-j)}=\E\{\hat{\bm{f}}_{\hat{\theta}^{(-j)}}^{(-j)}\vert\hat{\bm{f}}^{(-j)},\hat{\theta}^{(-j)}\}, \bm{v}^{(-j)}=\E\{\bm{f}_{\hat{\theta}^{(-j)}}\vert\hat{\theta}^{(-j)}\}, \bm{w}^{(-j)}=\E\{\bm{s}_{\hat{\theta}^{(-j)}}\vert\hat{\theta}^{(-j)}\},\\
    \bm{U}^{(-j)}=(\hat{\bm{f}}_{\hat{\theta}^{(-j)}}^{(-j)}-\bm{u}^{(-j)})(\hat{\bm{f}}_{\hat{\theta}^{(-j)}}^{(-j)}-\bm{u}^{(-j)})^\top, \bm{V}^{(-j)}=(\bm{f}_{\hat{\theta}^{(-j)}}-\bm{v}^{(-j)})(\bm{f}_{\hat{\theta}^{(-j)}}-\bm{v}^{(-j)})^\top,\\
    \bm{U}_s^{(-j)}=(\bm{s}_{\hat{\theta}^{(-j)}}-\bm{w}^{(-j)})(\hat{\bm{f}}_{\hat{\theta}^{(-j)}}^{(-j)}-\bm{u}^{(-j)})^\top, \bm{V}_s^{(-j)}=(\bm{s}_{\hat{\theta}^{(-j)}}-\bm{w}^{(-j)})(\bm{f}_{\hat{\theta}^{(-j)}}-\bm{v}^{(-j)})^\top.
\end{gather*}

\begin{asmp}\label{asmp:cross T moment}
    For every $j\in\{1,\dots,K\}$, $(k,l)\in \{1,\dots,q\}$, and some $r>1$, $c<\infty$,
    \begin{gather*}
        \E\{\vert\bm{U}_{(kl)}^{(-j)}-\bm{V}_{(kl)}^{(-j)}\vert^r\} < c, \E\{\bm{U}_{(kl)}^{(-j)}-\bm{V}_{(kl)}^{(-j)}\}=o(1),\\
        \E\{\vert\bm{U}_{s,(kl)}^{(-j)}-\bm{V}_{s,(kl)}^{(-j)}\vert^r\} < c, \E\{\bm{U}_{s,(kl)}^{(-j)}-\bm{V}_{s,(kl)}^{(-j)}\}=o(1).
    \end{gather*}
\end{asmp}
\begin{asmp}\label{asmp:cross f moment}
$\E\{\Vert\hat{\bm{f}}_{\hat{\theta}^{(-j)}}^{(-j)}-\bm{f}_{\hat{\theta}^{(-j)}}\Vert^2\}=o(1)$ for any $j\in\{1,\dots,K\}$.    
\end{asmp}
Assumption \ref{asmp:cross T moment} is a moment condition which ensures the consistency of the matrix $\widehat{\bm{T}}^{(j)}_{ss}(\hat{\theta}^{(-j)})$. Assumption \ref{asmp:cross f moment} is a stability condition, and its similar version can be found in \cite{zrnic2024cross} and \cite{kim2024semi}. It is worth pointing out that here we have an additional stochastic term $\hat{\theta}^{(-j)}$ and the expectation is with respect to $(\mathbf{X},\hat{\bm{f}}^{(-j)},\hat{\theta}^{(-j)})$. 
To prove Theorem \ref{thm:cross pdc}, we need the following two useful lemmas.
\begin{lem}\label{lem:conditional O_p(1)}[Lemma S2 in \cite{zhang2022high}]
    Let $\{W_n\}$ be a sequence of random variables, $\{\mathcal{F}_n\}$ be a sequence of $\sigma$-algebra. If $\E(W_n^2\vert \mathcal{F}_n)=O_p(1)$, then $W_n=O_p(1)$. Consequently, if $(Z_{n,i})$ is a row-wise independent and identically distributed triangular array conditional on $\mathcal{F}_n$, with $\var(Z_{n,1}\vert \mathcal{F}_n)=O_p(1)$, or a stronger condition that $\E(Z_{n,1}^2\vert \mathcal{F}_n)=O_p(1)$. Then $\frac{1}{n}\sum_{i=1}^n Z_{n,i}=\E(Z_{n,1})+O_p(n^{-1/2})$.
\end{lem}

\begin{lem}\label{lem:conditional o_p(1)}
    Let $\{\mathcal{F}_n\}$ be a sequence of $\sigma$-algebra. $\{Z_{n,i}\}$ is a row-wise independent and identically distributed triangular
array conditional on $\sigma$-algebra $\mathcal{F}_n$ with $\E (Z_{n,1})=0$ and $\E\vert Z_{n,1} \vert^r < c $ for some $r>1$ and $c<\infty$. Let $W_n=\frac{1}{n}\sum_{i=1}^n Z_{n,i}$, then $W_n=o_p(1)$.  
\end{lem}

\begin{proof}[Proof of Lemma \ref{lem:conditional o_p(1)}]
    This lemma can be proved according to a similar strategy to the proof of Lemma S3 in \cite{zhang2022high}. For completeness, we present the proof.  Let $Z_{n,i}'=Z_{n,i}\bm{1}\{\vert Z_{n,i}\vert<n\}$, for any $C>0$, we have
    \begin{gather*}
        \Pr(\vert W_n\vert>C)\le \Pr(\vert\sum_{i=1}^n Z_{n,i}'\vert>nC)+\Pr(\cup_{i=1}^n\{\vert Z_{n,i}\vert\ge n\}).
    \end{gather*}
    By Markov's Inequality, 
    \begin{gather*}
        \Pr(\cup_{i=1}^n\{\vert Z_{n,i}\vert\ge n\})\le n\Pr(\vert Z_{n,1}\vert\ge n)\le n \frac{\E\vert Z_{n,1}\vert^r}{n^r}<cn^{1-r}=o(1).
    \end{gather*}
    If $r\ge 2$, then 
    \begin{gather*}
    \Pr(\vert\sum_{i=1}^n Z_{n,i}'\vert>nC)\le \frac{\E\vert\sum_{i=1}^n Z_{n,i}'\vert^2}{n^2 C^2}=\frac{\E\{\E\{\vert\sum_{i=1}^n Z_{n,i}'\vert^2\vert\mathcal{F}_n\}\}}{n^2C^2}\\
    =\frac{\E\{\E\{\vert Z_{n,1}'\vert^2\vert\mathcal{F}_n\}\}}{nC^2}=\frac{\E\{\vert Z_{n,i}\vert^{1/2} \vert Z_{n,i}\vert^{3/2} \bm{1}\{\vert Z_{n,i}\vert<n\}\}}{nC^2}<\frac{\E \vert Z_{n,i}\vert^{3/2}}{n^{1/2}C^2}=o(1).
    \end{gather*}
    If $1<r<2$, by the similar procedure, we have  
    \begin{equation*}
        \Pr(\vert\sum_{i=1}^n Z_{n,i}'\vert>nC)\le \frac{\E\{\vert Z_{n,i}\vert^{2-r} \vert Z_{n,i}\vert^{r} \bm{1}\{\vert Z_{n,i}\vert<n\}\}}{nC^2}<\frac{\E \vert Z_{n,i}\vert^{r}}{n^{r-1}C^2}=o(1).
    \end{equation*}
    Thus, we have proved that $W_n=o_p(1)$.
\end{proof}

\begin{proof}[Proof of Theorem \ref{thm:cross pdc}]
    Define
    \begin{equation*}
        \hat{\theta}_{cs,1}=\frac{1}{K}\sum_{j=1}^K\hat{\theta}^{(-j)}-\widehat{\bm{H}}^{-1}\hat{S}_{cs},
    \end{equation*}
    where $\hat{S}_{cs}=1/K\sum_{j=1}^K\E_{\mathcal{D}_L^{(j)}}\Big[\bm{s}_{\hat{\theta}^{(-j)}}+\gamma\widehat{\bm{T}}^{(j)}(\hat{\theta}^{(-j)})\big\{\bm{f}_{\hat{\theta}^{(-j)}}-\E_{\mathcal{D}_{U}^{(j)}}(\bm{f}_{\hat{\theta}^{(-j)}})\big\}\Big]$. We first show that 
    \begin{equation}\label{equ:cs asmp dist}
        \sqrt{n}(\hat{\theta}_{cs,1}-\theta^*)\overset{D}{\rightarrow} \N_d(\bm{0},\bm{H}^{-1}\bm{\Gamma}(\gamma)\{\bm{H}^\top\}^{-1}).
    \end{equation}
    As in the proof of Theorem \ref{thm:one-step pdc}, we decompose $\sqrt{n}\widehat{\bm{H}}(\hat{\theta}_{cs,1}-\theta^*)$ as 
    \begin{align*}
        \sqrt{n}\widehat{\bm{H}}(\hat{\theta}_{cs,1}-\theta^*)&=\frac{1}{K}\sum_{j=1}^K\Big[\sqrt{n}\widehat{\bm{H}}(\hat{\theta}^{(-j)}-\theta^*)-\sqrt{n}\E_{\mathcal{D}_L^{(j)}}(\bm{s}_{\hat{\theta}^{(-j)}}-\bm{s}_{\theta^*})\\
        &\quad-\sqrt{n}\big[\E_{\mathcal{D}_L^{(j)}}(\bm{s}_{\theta^*})+\gamma\bm{T}^*(\theta^*)\{\E_{\mathcal{D}_L^{(j)}}(\bm{f}_{\theta^*})-\E_{\mathcal{D}_U^{(j)}}(\bm{f}_{\theta^*})\}\big]\\
        &\quad-\sqrt{n}\gamma\bm{T}^*(\theta^*)\big\{\E_{\mathcal{D}_L^{(j)}}(\bm{f}_{\hat{\theta}^{(-j)}}-\bm{f}_{\theta^*})-\E_{\mathcal{D}_U^{(j)}}(\bm{f}_{\hat{\theta}^{(-j)}}-\bm{f}_{\theta^*})\big\}\\
        &\quad-\sqrt{n}\gamma\big\{\widehat{\bm{T}}^{(j)}(\hat{\theta}^{(-j)})-\bm{T}^*(\theta^*)\big\}\big\{\E_{\mathcal{D}_L^{j}}(\bm{f}_{\hat{\theta}^{(-j)}})-\E_{\mathcal{D}_U^{j}}(\bm{f}_{\hat{\theta}^{(-j)}})\big\} \Big]\\
        &:=\frac{1}{K}\sum_{j=1}^K\{\RN{1}^{(j)}-\RN{2}^{(j)}-\RN{3}^{(j)}-\RN{4}^{(j)}\}.
    \end{align*}
    Since $\hat{\theta}^{(-j)}$ is a $\sqrt{n}$-consistent estimator of $\theta^*$ for every $j\in\{1,\dots,K\}$, following the same procedure as in the proof of Theorem \ref{thm:one-step pdc}, we have
    \begin{equation*}
        \RN{1}^{(j)}=o_p(1),\quad \RN{3}^{(j)}=o_p(1),\quad \RN{4}^{(j)}=o_p(1),\quad \text{for every}\quad  j\in\{1,\dots,K\}. 
    \end{equation*}
    Given that $K$ is fixed, $\frac{1}{K}\sum_{j=1}^K\{\RN{1}^{(j)}-\RN{3}^{(j)}-\RN{4}^{(j)}\}=o_p(1)$. Note that $\frac{1}{K}\sum_{j=1}^K \RN{2}^{(j)}=\RN{2}$, and using (\ref{equ:II asymp normal}), we have proved that (\ref{equ:cs asmp dist}) holds.

    Next, we show that $\hat{S}_{ss}-\hat{S}_{cs}=o_p(n^{-1/2})$, where $\hat{S}_{ss}=1/K\sum_{j=1}^K \hat{S}^{(j)}$. Decompose $\hat{S}_{ss}-\hat{S}_{cs}$ as follows:
    \begin{align*}
        &\hat{S}_{ss}-\hat{S}_{cs}\\
        &=\gamma\frac{1}{K}\sum_{j=1}^K\E_{\mathcal{D}_{L}^{(j)}}\{\widehat{\bm{T}}^{(j)}_{ss}(\hat{\theta}^{(-j)})\hat{\bm{f}}^{(-j)}_{\hat{\theta}^{(-j)}}-\widehat{\bm{T}}^{(j)}(\hat{\theta}^{(-j)})\bm{f}_{\hat{\theta}^{(-j)}}\}\\
        &\quad-\gamma\frac{1}{K}\sum_{j=1}^{K}\E_{\mathcal{D}_{U}^{(j)}}\{\widehat{\bm{T}}^{(j)}_{ss}(\hat{\theta}^{(-j)})\hat{\bm{f}}^{(-j)}_{\hat{\theta}^{(-j)}}-\widehat{\bm{T}}^{(j)}(\hat{\theta}^{(-j)})\bm{f}_{\hat{\theta}^{(-j)}}\}\\
        &=\gamma\frac{1}{K}\sum_{j=1}^K \Big\{\widehat{\bm{T}}^{(j)}(\hat{\theta}^{(-j)})\big\{\E_{\mathcal{D}_{L}^{(j)}}(\hat{\bm{f}}^{(-j)}_{\hat{\theta}^{(-j)}}-\bm{f}_{\hat{\theta}^{(-j)}})-\E_{\mathcal{D}_U^{(j)}}(\hat{\bm{f}}^{(-j)}_{\hat{\theta}^{(-j)}}-\bm{f}_{\hat{\theta}^{(-j)}})\big\}\\
    &\quad+\big\{\widehat{\bm{T}}^{(j)}_{ss}(\hat{\theta}^{(-j)})-\widehat{\bm{T}}^{(j)}(\hat{\theta}^{(-j)})\big\}\big\{\E_{\mathcal{D}_L^{(j)}}(\hat{\bm{f}}_{\hat{\theta}^{(-j)}}^{(-j)})-\E_{\mathcal{D}_U^{(j)}}(\hat{\bm{f}}_{\hat{\theta}^{(-j)}}^{(-j)})\big\}\Big\}\\
    &:=\gamma\frac{1}{K}\sum_{j=1}^K \{M_1^{(j)}+M_2^{(j)}\},
    \end{align*}
We next show that, for any $j\in\{1,\dots,K\}$, $\widehat{\bm{T}}^{(j)}_{ss}(\hat{\theta}^{(-j)})-\widehat{\bm{T}}^{(j)}(\hat{\theta}^{(-j)})=o_p(1)$. Note that 
\begin{align*}
    &\cov_{\mathcal{D}_L^{(j)}}(\hat{\bm{f}}_{\hat{\theta}^{(-j)}}^{(-j)})-\cov_{\mathcal{D}_L^{(j)}}(\bm{f}_{\hat{\theta}^{(-j)}})\\
    &=\E_{\mathcal{D}_{L}^{(j)}}(\bm{U}^{(-j)}-\bm{V}^{(-j)})-\Big\{\E_{\mathcal{D}_{L}^{(j)}}(\hat{\bm{f}}_{\hat{\theta}^{(-j)}}^{(-j)})-\bm{u}^{(-j)}\Big\}\Big\{\E_{\mathcal{D}_{L}^{(j)}}(\hat{\bm{f}}_{\hat{\theta}^{(-j)}}^{(-j)})-\bm{u}^{(-j)}\Big\}^\top\\
    &\quad+\Big\{\E_{\mathcal{D}_{L}^{(j)}}(\bm{f}_{\hat{\theta}^{(-j)}})-\bm{v}^{(-j)}\Big\}\Big\{\E_{\mathcal{D}_{L}^{(j)}}(\bm{f}_{\hat{\theta}^{(-j)}})-\bm{v}^{(-j)}\Big\}^\top.
\end{align*}
By Assumptions \ref{asmp:score is L2 bounded}, \ref{asmp:cross T moment}, \ref{asmp:cross f moment} and Lemma \ref{lem:conditional o_p(1)}, we have 
\begin{gather*}
    \E_{\mathcal{D}_{L}^{(j)}}\{\bm{U}^{(-j)}-\bm{V}^{(-j)}\}=o_p(1),\\
    \E_{\mathcal{D}_L^{(j)}}\{\hat{\bm{f}}_{\hat{\theta}^{(-j)}}^{(-j)}-\bm{u}^{(-j)}\}=o_p(1), \quad\E_{\mathcal{D}_L^{(j)}}\{\bm{f}_{\hat{\theta}^{(-j)}}-\bm{v}^{(-j)}\}=o_p(1).
\end{gather*}
Since the elements of the matrix inverse are continuous functions of the original matrix elements, which implies that
$\cov^{-1}_{\mathcal{D}_L^{(j)}}(\hat{\bm{f}}_{\hat{\theta}^{(-j)}}^{(-j)})-\cov^{-1}_{\mathcal{D}_L^{(j)}}(\bm{f}_{\hat{\theta}^{(-j)}})=o_p(1)$. Following the similar procedure, we can prove that $\cov_{\mathcal{D}_l^{(j)}}(\bm{s}_{\hat{\theta}^{(-j)}},\hat{\bm{f}}_{\hat{\theta}^{(-j)}}^{(-j)})-\cov_{\mathcal{D}_l^{(j)}}(\bm{s}_{\hat{\theta}^{(-j)}},\bm{f}_{\hat{\theta}^{(-j)}})=o_p(1)$. Thus $\widehat{\bm{T}}^{(j)}_{ss}(\hat{\theta}^{(-j)})-\widehat{\bm{T}}^{(j)}(\hat{\theta}^{(-j)})=o_p(1)$. 

Moreover, by Assumptions \ref{asmp:consist init}, \ref{asmp:score is L2 bounded} and \ref{asmp:smooth score}, 
\begin{small}
\begin{gather*}
    \E\{\Vert\bm{f}_{\hat{\theta}^{(-j)}}-\bm{f}_{\theta^*}\Vert^2\}\le \E\{\dot{f}^2(\mathbf{X})\Vert\hat{\theta}^{(-j)}-\theta^*\Vert^2\}=\E\{\dot{f}^2(\mathbf{X})\}\E\{\Vert\hat{\theta}^{(-j)}-\theta^*\Vert^2\}=o(n^{-1}).
\end{gather*}
\end{small}
This, combined with Assumption \ref{asmp:cross f moment} yields that $\E\{\Vert\hat{\bm{f}}_{\hat{\theta}^{(-j)}}^{(-j)}-\bm{f}_{\theta^*}\Vert^2\}=o(1)$. By triangular inequality and Assumption \ref{asmp:score is L2 bounded}, we have $\E\{\Vert\hat{\bm{f}}_{\hat{\theta}^{(-j)}}^{(-j)}\Vert^2\}=O(1)$. Moreover, Markov inequality implies that $\E\{\Vert\hat{\bm{f}}_{\hat{\theta}^{(-j)}}^{(-j)}\Vert^2\vert \mathcal{F}_n\}=O_p(1)$. Then, by Lemma \ref{lem:conditional O_p(1)}, we have $\E_{\mathcal{D}_L^{(j)}}(\hat{\bm{f}}_{\hat{\theta}^{(-j)}}^{(-j)})=\E(\hat{\bm{f}}_{\hat{\theta}^{(-j)}}^{(-j)})+ O_p(n^{-1/2})$. Therefore, for every $j\in\{1,\dots,K\}$, 
\begin{equation}\label{equ:cross second term}
    M_2^{(j)}=o_p(n^{-1/2}).
\end{equation}

 On the other hand, let $\Delta_{\hat{\bm{f}}^{(-j)},\hat{\theta}^{(-j)}}=\E\{\hat{\bm{f}}_{\hat{\theta}^{(-j)}}^{(-j)}-\bm{f}_{\hat{\theta}^{(-j)}}\vert\hat{\bm{f}}^{(-j)},\hat{\theta}^{(-j)}\}$, we have 
    \begin{align*}
        &\E\{\Vert \E_{\mathcal{D}_{L}^{(j)}}(\hat{\bm{f}}_{\hat{\theta}^{(-j)}}^{(-j)}-\bm{f}_{\hat{\theta}^{(-j)}}-\Delta_{\hat{\bm{f}}^{(-j)},\hat{\theta}^{(-j)}}) \Vert^2\}\\
        &=\E\{\E\{\Vert \E_{\mathcal{D}_{L}^{(j)}}(\hat{\bm{f}}_{\hat{\theta}^{(-j)}}^{(-j)}-\bm{f}_{\hat{\theta}^{(-j)}}-\Delta_{\hat{\bm{f}}^{(-j)},\hat{\theta}^{(-j)}}) \Vert^2\big\vert\hat{\bm{f}}^{(-j)},\hat{\theta}^{(-j)}\}\}\\
        &=\frac{K}{n}\E\{\E\{\Vert\hat{\bm{f}}_{\hat{\theta}^{(-j)}}^{(-j)}-\bm{f}_{\hat{\theta}^{(-j)}}-\Delta_{\hat{\bm{f}}^{(-j)},\hat{\theta}^{(-j)}}\Vert^2\big\vert\hat{\bm{f}}^{(-j)},\hat{\theta}^{(-j)}\}\}\\
        &\le\frac{K}{n}\E\{\E\{\Vert\hat{\bm{f}}_{\hat{\theta}^{(-j)}}^{(-j)}-\bm{f}_{\hat{\theta}^{(-j)}}\Vert^2\big\vert\hat{\bm{f}}^{(-j)},\hat{\theta}^{(-j)}\}\}\\
        &=\frac{K}{n}\E\{\Vert\hat{\bm{f}}_{\hat{\theta}^{(-j)}}^{(-j)}-\bm{f}_{\hat{\theta}^{(-j)}}\Vert^2\}=o(n^{-1}),
    \end{align*}
    where the first and fourth steps follow from law of total expectation, the second step follows by the conditional independence. Hence $\E_{\mathcal{D}_{L}^{(j)}}\{\hat{\bm{f}}_{\hat{\theta}^{(-j)}}^{(-j)}-\bm{f}_{\hat{\theta}^{(-j)}}-\Delta_{\hat{\bm{f}}^{(-j)},\hat{\theta}^{(-j)}}\}=o_p(n^{-1/2})$. This combined with (\ref{equ:T diff=o_p(1)}) implies that, for every $j\in\{1,\dots,K\}$, 
    \begin{equation}\label{equ:cross first term}
        M_1^{(j)}=o_p(n^{-1/2}).
    \end{equation}

    Since $K$ is fixed, combining (\ref{equ:cross second term}) with (\ref{equ:cross first term}) leads to $\hat{S}_{ss}-\hat{S}_{cs}=o_p(n^{-1/2})$. Therefore, 
    \begin{gather*}
        \sqrt{n}(\hat{\theta}_{cs,1}-\hat{\theta}_{sspdc,1})=\sqrt{n}\widehat{\bm{H}}^{-1}(\hat{S}_{ss}-\hat{S}_{cs})=o_p(1).
    \end{gather*}
    Finally, by Slutsky's theorem and (\ref{equ:cs asmp dist}), we conclude that  (\ref{equ:cspdc aymp normal}) holds.
\end{proof}

\begin{proof}[Proof of Theorem \ref{thm:random f}]
         Under Assumption \ref{asmp:stochastic asmp}(\ref{asmp:smooth stochastic f}), Lemma 19.24 of \cite{vaart_1998} implies
        \begin{equation*}
            \mathbb{G}_n(\tilde{\bm{f}}_{\hat{\theta}_0}-\tilde{\bm{f}}_{\theta^*})=o_p(1).
        \end{equation*}
        Under Assumption \ref{asmp:stochastic asmp}(\ref{asmp:f decompose}),
        \begin{gather*}
            \E_n(\bm{f}_{\zeta;\hat{\theta}_0}-\bm{f}_{\zeta;\theta^*})=\E_n(\tilde{\bm{f}}_{\hat{\theta}_0}-\tilde{\bm{f}}_{\theta^*})=\E(\tilde{\bm{f}}_{\hat{\theta}_0}-\tilde{\bm{f}}_{\theta^*})+n^{-1/2}\mathbb{G}_n(\tilde{\bm{f}}_{\hat{\theta}_0}-\tilde{\bm{f}}_{\theta^*})\\
            =\E(\tilde{\bm{f}}_{\hat{\theta}_0}-\tilde{\bm{f}}_{\theta^*})+o_p(n^{-1/2}).
        \end{gather*}
	    Similar to equation \eqref{equ:one-step decompose}, we can decompose $\sqrt{n}\widehat{\bm{H}}(\hat{\theta}_{pdc,1}-\theta^*)$ as 
	\begin{align*}
		\sqrt{n}\widehat{\bm{H}}(\hat{\theta}_{pdc,1}-\theta^*)&=\RN{1}-\RN{2}\\
		&\quad -\sqrt{n}\gamma\bm{T}^*(\theta^*)\big\{\E_n(\tilde{\bm{f}}_{\hat{\theta}_0}-\tilde{\bm{f}}_{\theta^*})-\E_N(\tilde{\bm{f}}_{\hat{\theta}_0}-\tilde{\bm{f}}_{\theta^*})\big\}\\
		&\quad-\sqrt{n}\gamma\{\widehat{\bm{T}}(\hat{\theta}_0)-\bm{T}^*(\theta^*)\}\big\{\E_n(\tilde{\bm{f}}_{\hat{\theta}_0})-\E_{N}(\tilde{\bm{f}}_{\hat{\theta}_0})\big\}\\
        &:=\RN{1}-\RN{2}-\tilde{\RN{3}}-\tilde{\RN{4}}.
	\end{align*}
	 The proof follows the same procedure as that of the Theorem \ref{thm:one-step pdc}. 
\end{proof}

\begin{proof}[Proof of Corollary \ref{coro:estimated optimal gamma}]
	We denote the one-step PDC estimator when $\gamma=\gamma_n$ as $\hat{\theta}_{pdc,1}(\gamma_n)$.  We can still use the decomposition in equation \eqref{equ:one-step decompose}.  We brief denote the decomposition as $\sqrt{n}\hat{\bm{H}}(\hat{\theta}_{pdc,1}(\gamma_n)-\theta^*)=\RN{1}-\RN{2}_n-\RN{3}_n-\RN{4}_n$. 
	$\RN{1}$ is the same as the one in the fixed $\gamma$ case.  Since $\gamma_n$ is bounded, we still have $\RN{3}_n=o_p(1)$ and $\RN{4}_n=o_p(1)$. For $\RN{2}_n$, we have
	\begin{align*}
		\RN{2}_n-\RN{2}=(\gamma_n-\gamma_{opt})\sqrt{n}\bm{T}^*(\theta^*)\{\E_n(\bm{f}_{\theta^*})-\E_N(\bm{f}_{\theta^*})\}=(\gamma_n-\gamma_{opt})O_p(1)=o_p(1).
	\end{align*}
	Thus we have completed the proof.
\end{proof}

\noindent\textbf{Additional results for real data analysis}

We present the estimated $\bm{T}$ of the PDC and the tuning parameter $\lambda$ of the PPI++ in the analysis of the homeless dataset as follows.
\begin{small}
\begin{gather*}
    \widehat{\bm{T}}(\hat{\theta}_0)=
    \begin{bmatrix}
        0.547    &    0.001    &     0.004 &  -0.007  &     -0.005 &    -0.007  &   -0.002       &          0.002\\
         -0.942   &     0.111     &    0.036  &  0.040   &     0.021   &   0.016   &  -0.043        &         0.038\\
         21.196   &    -0.068     &    0.192  & -0.061    &   -0.186   &  -0.324    &  0.025      &           0.052\\
              7.089   &    -0.002     &    0.035 &  -0.237   &    -0.060   &  -0.089   &  -0.007   &              0.016\\
          26.999   &     0.004     &    0.066  & -0.294   &    -0.353   &  -0.332   &  -0.064       &          0.086\\
            12.178  &      0.117    &     0.139  & -0.344    &   -0.047   &  -0.085  &   -0.103       &          0.005\\
            29.379   &     0.077     &    0.343  & -0.107    &   -0.314   &  -0.433  &   -0.081         &        0.059\\
 22.267    &    0.081   &      0.119  & -0.380   &    -0.082   &  -0.108  &   -0.158    &            -0.084
    \end{bmatrix},\\
    \hat{\lambda}(\hat{\theta}_0)=-0.060.
\end{gather*}
\end{small}
\newpage

 \end{document}